\PassOptionsToPackage{nosumlimits,nonamelimits}{amsmath}
\PassOptionsToPackage{capitalize,noabbrev,nameinlink}{cleveref}
\PassOptionsToPackage{final}{graphicx}
\PassOptionsToPackage{final}{hyperref}

\documentclass[a4paper,UKenglish,cleveref, numberwithinsect,final]{lipics-v2021}

\hideLIPIcs

\title{Non-Expansive Quantitative Coalgebraic Logic}

\ccsdesc[100]{Theory of computation~Modal and temporal logics}
\ccsdesc[100]{Theory of computation~Description logics}
\ccsdesc[100]{Theory of computation~Automated reasoning}

\keywords{Quantitative modal logics, tableau algorithms, coalgebra, complexity}

\author{Stefan Gebhart}{Friedrich-Alexander-University Erlangen-N\"urnberg, Germany}{Stefan.Gebhart@fau.de}{https://orcid.org/0009-0000-4029-5729}{}

\author{Lutz Schr\"oder}{Friedrich-Alexander-University Erlangen-N\"urnberg, Germany}{Lutz.Schroeder@fau.de}{https://orcid.org/0000-0002-3146-5906}{}
\author{Paul Wild}{Friedrich-Alexander-University Erlangen-N\"urnberg, Germany}{Paul.Wild@fau.de}{https://orcid.org/0000-0001-9796-9675}{}

\authorrunning{S. Gebhart, L. Schr\"oder and P. Wild}

\funding{Funded by the Deutsche Forschungsgemeinschaft
  (DFG, German Research Foundation) -- project numbers 531706730, 434050016, 419850228} 

\relatedversion{This is the full version of a paper accepted at CSL 2027.}

\Copyright{Stefan Gebhart, Lutz Schr\"oder and Paul Wild}

\nolinenumbers

\EventEditors{Corina C\^{\i}rstea and Yannick Forster}
\EventNoEds{2}
\EventLongTitle{35th EACSL Annual Conference on Computer Science Logic (CSL 2027)}
\EventShortTitle{CSL 2027}
\EventAcronym{CSL}
\EventYear{2027}
\EventDate{January 26--29, 2027}
\EventLocation{Brighton, UK}
\EventLogo{}
\SeriesVolume{403}
\ArticleNo{20}

\usepackage{booktabs}

\usepackage{amsthm}
\usepackage{fancyvrb}
\usepackage{mathtools}
\usepackage{mathrsfs}
\usepackage{xspace}
\usepackage{makeidx}
\usepackage{stmaryrd}

\usepackage[linesnumbered,ruled,vlined]{algorithm2e}
\SetAlCapNameFnt{\normalsize}
\SetAlCapFnt{\normalsize}

\usepackage{stackengine}

\usepackage[layout=footnote,marginclue,final]{fixme}
\FXRegisterAuthor{sg}{asg}{SG}
\FXRegisterAuthor{ls}{als}{LS}
\FXRegisterAuthor{pw}{apw}{PW}

\theoremstyle{definition}
\newtheorem{defn}[theorem]{Definition}
\newtheorem{convention}[theorem]{Convention}
\crefname{defn}{definition}{definitions}
\Crefname{defn}{Definition}{Definitions}

\newcommand{\Pow}{\mathcal{P}}
\newcommand{\fPow}{\Pow_{[0,1]}}
\newcommand{\Dist}{\mathcal{D}}
\newcommand{\Rat}{\mathbb{Q}} \newcommand{\ALC}{\mathcal{ALC}}
\newcommand{\EXP}{\textsc{ExpTime}\xspace}

\newcommand{\NEXP}{\textsc{NExpTime}\xspace}
\newcommand{\PSPACE}{\textsc{PSpace}\xspace}
\newcommand{\NPSPACE}{\textsc{NPSpace}\xspace}
\newcommand{\EXPSPACE}{\textsc{EXPSpace}\xspace}
\newcommand{\NP}{\textsc{NP}\xspace}
\newcommand{\GENERALLY}{\boldsymbol{G}}
\newcommand{\PROBABLY}{\boldsymbol{E}}
\newcommand{\SET}{\mathsf{Set}}
\newcommand{\op}{\operatorname*{op}}

\newcommand{\Lgen}{\mathcal{L}_{\mathsf{gen}}}
\newcommand{\matches}{\mathsf{matches}}
\newcommand{\concls}{\mathsf{concls}}
\newcommand{\concl}{\mathsf{concl}}

\newcommand{\tin}{\mathrel{\mathchoice
		{\tinaux\displaystyle      {1.00ex}}
		{\tinaux\textstyle         {1.00ex}}
		{\tinaux\scriptstyle       {0.70ex}}
		{\tinaux\scriptscriptstyle {0.50ex}}
}}

\newcommand{\tinaux}[2]{
	\stackMath
	\stackinset{c}{0pt}{c}{#2}
	{#1\cdot}
	{#1\in}
}

\def\Xint#1{\mathchoice
	{\XXint\displaystyle\textstyle{#1}}
	{\XXint\textstyle\scriptstyle{#1}}
	{\XXint\scriptstyle\scriptscriptstyle{#1}}
	{\XXint\scriptscriptstyle\scriptscriptstyle{#1}}
	\!\int}
\def\XXint#1#2#3{{\setbox0=\hbox{$#1{#2#3}{\int}$}
		\vcenter{\hbox{$#2#3$}}\kern-.5\wd0}}
\def\dashint{\Xint-}

\begin{document}

\maketitle

\begin{abstract}
  Quantitative (or \emph{fuzzy}) logic extends the classical truth values ``true'' and ``false'' with additional truth degrees, typically real numbers in the unit interval. More specifically, quantitative modal logics in this sense are given by a choice of quantitative modalities and a quantitative propositional base. It has been noted that quantitative modal logics over the Zadeh base, which interprets disjunction as maximum, are often computationally tractable but on the other hand add little in the way of expressivity over their classical counterparts. Contrastingly, quantitative modal logics over the more expressive \L{}ukasiewicz base have attractive logical properties but are often computationally less tractable or even undecidable. In the basic case of the modal logic of quantitative relations, sometimes termed \emph{fuzzy $\ALC$}, it has recently been shown that an intermediate \emph{non-expansive} propositional base, known from characteristic logics for behavioural distances of quantitative systems, strikes a balance between these poles: It provides increased expressiveness over the Zadeh base while avoiding the computational problems of the \L{}ukasiewicz base, in fact allowing for reasoning in \PSPACE. Modal logics, in particular quantitative modal logics, generally vary widely in terms of syntax and semantics, involving, for instance, probabilistic, preferential, or weighted structures. Coalgebraic modal logic provides a unifying framework for wide ranges of semantically different modal logics, both two-valued and quantitative. In the present work, we focus on \emph{non-expansive coalgebraic quantitative modal logics}, providing a criterion for decidability in \PSPACE. Using this criterion, we recover the mentioned complexity result for non-expansive fuzzy $\ALC$ and moreover obtain new \PSPACE upper bounds for quantitative modal logics over probabilistic transition systems and for metric modal logics under compactness and computability assumptions on the label space. Notably, we show that the logic of \emph{generally}, which has recently been shown to characterize $\epsilon$-distance (or L\'evy-Prokhorov distance) on Markov chains, is decidable in \PSPACE.
\end{abstract}

\section{Introduction}\label{intro}

Quantitative logics with real-valued truth degrees in the unit
interval, widely known as \emph{fuzzy logics}
(e.g.~\cite{ZadehAliev18}), offer a more fine-grained notion of truth
than two-valued classical logics. They have been popularized in
knowledge representation as providing expressive means for vague
real-world phenomena such as tallness or mutual dislike. Beyond basic
propositional logics, quantitative \emph{modal} logics have
correspondingly seen applications in their incarnation as fuzzy
description logics (e.g.~\cite{LukasiewiczStraccia08}), where, for
instance, people who like only tall people would be described by the
concept $\forall\,\mathsf{likes}.\,\mathsf{tall}$. Quantitative modal
logics additionally play a role in concurrency and model checking. A
well-known example is the probabilistic
$\mu$-calculus~\cite{MorganMcIver97,HuthKwiatkowska97}, in which
quantitative truth values are combined using fixpoints, propositional
operators, and a probabilistic expectation modality. Various next-step
modal logics have moreover appeared in Hennessy-Milner-type
characterization theorems stating the coincidence of logical distance and
various forms of \emph{behavioural distance}. The prototypical result
of this kind is the coincidence of a behavioural distance on
probabilistic transition systems defined using the Kantorovich
distance between distributions on the one hand and logical distance in a
quantitative probabilistic modal logic based on an expectation
modality on the other hand~\cite{BreugelWorrell05}. Similar results
have been obtained for various types of quantitative relational
systems (e.g.~\cite{AlfaroEA09,Fan15}) as well as for systems combining
probability and non-determinism~\cite{EleftheriouEA12,DuEA16}, and
moreover have been proved in coalgebraic
generality~\cite{KonigMikaMichalski18,WildSchroder22}, thus covering
also weighted, neighbourhood-based, and other system types.

Our main concern in the present work is automated reasoning in quantitative
modal logics. By the \emph{satisfiability problem} for our logics we mean deciding
whether some state in some model satisfies a given finite set of bounds on formula
values. Strict and non-strict lower and upper bounds with rational thresholds may
be mixed in the same sequent. In particular, one may ask whether a single formula
can take a value greater than~$p$ or, respectively, at least~$p$.
The complexity of the central reasoning problems in
quantitative modal logics depends strongly on the underlying propositional
base. Two well-known poles in the landscape of quantitative propositional
systems are, on the one hand, the \emph{Zadeh} base, which features only
minimum, maximum, and negation $x\mapsto 1-x$; and, on the other
hand, the more expressive \emph{\L{}ukasiewicz} base, essentially given
by interpreting disjunction as truncated addition. In the basic case
of quantitative-relational modal logic (in description logic parlance,
\emph{fuzzy $\ALC$}~\cite{LukasiewiczStraccia08}), modal logic over
the Zadeh base has moderate complexity, but on the other hand in fact
essentially coincides with two-valued modal logic in the sense that a
formula is satisfiable with some threshold truth degree $p > 0.5$ iff
it is classically
satisfiable~\cite{Straccia01,KellerHeymans09}. Hence, satisfiability
with threshold~$p$ is (only) \PSPACE-complete (while satisfiability
with any threshold $p\le 0.5$ is decidable in linear time by
straightforward recursion over the formula
syntax~\cite{BonattiTettamanzi03}). Contrastingly, fuzzy $\ALC$ over
the \L{}ukasiewicz base is more expressive but has less favourable
computational properties: Reasoning under global assumptions (i.e.\
under a TBox, in description logic terms) is undecidable, and the best
known algorithms for satisfiability checking in the absence of global
assumptions take non-deterministic exponential
time~\cite{Straccia05,StracciaBobillo07,SchroderPattinson11,KulackaEA13}.

Many of the above-mentioned characteristic modal logics for
behavioural distances (including quantitative probabilistic modal
logic)~\cite{BreugelWorrell05,KonigMikaMichalski18,WildSchroder22} in
fact work with an intermediate \emph{non-expansive} propositional
base, which extends the Zadeh base with constant shifts (that is,
restricts \L{}ukasiewicz disjunction by requiring one disjunct to be
constant); this is owed precisely to the fact that the \L{}ukasiewicz
base contains operations that increase distance, such as addition,
while the Zadeh base in general does not suffice to characterize
behavioural distance. The non-expansive base is moreover employed in
the recently introduced fuzzy description logic \emph{non-expansive
  fuzzy $\ALC$}~\cite{ijcai2025p502}, which, unlike Zadeh fuzzy
$\ALC$, supports the specification of actual quantitative effects in
knowledge representation (such as the degree of football fandom being
passed on among friends with a constant decrease). Nevertheless,
non-expansive fuzzy $\ALC$ still allows reasoning in \PSPACE (and in
\EXP under global assumptions). Since non-expansiveness is the
quantitative analogue of bisimulation invariance~\cite{WildEA18}, this
confirms the slogan that the tractability of modal logics is owed
largely to their bisimulation invariance, which, for instance,
classically entails a tree model property~\cite{Vardi96}.  In a
nutshell, the contribution of the present work is to lift this result
to the level of generality of quantitative coalgebraic
logic~\cite{SchroderPattinson11,KonigMikaMichalski18,WildSchroder22}.
Generally, coalgebraic logic provides a unified treatment of the
syntax, semantics and algorithmics of wide ranges of modal
logics~\cite{Schroder07,SchroderPattinson09}, such as
probabilistic~\cite{LarsenSkou91}, game-based~\cite{Pauly02,AlurEA02},
and neighbourhood-based~\cite{Chellas80} logics. Its generality is
based on abstracting systems as coalgebras for a given functor
determining the system type following the paradigm of universal
coalgebra~\cite{Rutten00}, and modalities as predicate liftings for
the given functor~\cite{Pattinson04,Schroder08}. Quantitative modal
logics are specifically captured using liftings of
unit-interval-valued predicates to the given
functor~\cite{SchroderPattinson11,KonigMikaMichalski18,WildSchroder22}. We
develop a criterion for a quantitative coalgebraic modal logic over
the non-expansive propositional base in this sense to allow
satisfiability checking in \PSPACE, reducing the proof work to
properties that are fairly straightforward to check in concrete
instances. As one such instance, we recover the \PSPACE decidability
of non-expansive fuzzy
$\ALC$~\cite{ijcai2025p502,Hermes23}. Additionally, we obtain new
instances. In particular, we newly establish \PSPACE decidability of
modal logics of metric transition systems with
two-valued~\cite{AlfaroEA09} and quantitative~\cite{ForsterEA25}
transition relations under compactness and computability assumptions
on the label space, and of quantitative probabilistic modal logics
that relate to $\epsilon$-bisimilarity on Markov
chains~\cite{DesharnaisEA08}. Most notably, this concerns the
non-expansive logic of `generally', a variant of the quantitative
modality `probably' that has been introduced in the context of vague
knowledge representation~\cite{SchroderPattinson11}.  While `probably'
is just the expectation modality~\cite{BreugelWorrell05}, and as such
is essentially based on standard integration (or weighted sums in the
discrete setting), the `generally' modality instead employs Sugeno
integration as used in fuzzy measure theory~\cite{1975TheoryOF}. It
has recently been shown that in the same way as `probably' induces the
Kantorovich distance on distributions in the manner recalled above,
`generally' induces the popular L\'evy-Prokhorov
distance~\cite{WildEA26}, which in turn has recently been
shown~\cite{DesharnaisSokolova26} to induce $\epsilon$-bisimulation
distance on Markov chains~\cite{DesharnaisEA08}. By general results in
quantitative coalgebraic
logic~\cite{KonigMikaMichalski18,WildSchroder22}, this implies that
the modal logic of `generally' characterizes $\epsilon$-bisimulation
distance on finitely branching Markov chains. While the best known
upper bound for strict satisfiability in the logic of `probably' remains the $\EXPSPACE$ upper
bound established for the \L{}ukasiewicz
variant~\cite{SchroderPattinson11}, our results establish $\PSPACE$
decidability of the non-expansive logic of `generally'.

\subparagraph*{Related Work}\label{relwork} The use of (rational)
truth constants in \L{}ukasiewicz-type quantitative logics goes back
to (rational) Pavelka logic~\cite{Hajek95,Pavelka79}.  Constraining
\L{}ukasiewicz fuzzy $\ALC$ to finitely many truth values ensures
decidability of threshold satisfiability in
$\PSPACE$~\cite{BouEA11}. As mentioned above, only an $\NEXP$ upper
bound is known for infinite-valued \L{}ukasiewicz fuzzy
$\ALC$~\cite{Straccia05,SchroderPattinson11,KulackaEA13}. The same
holds for our probabilistic instance logics (the logic of
\emph{generally} and a variant with explicit thresholds, termed
\emph{quantitative fuzzy $\ALC$}), for whose \L{}ukasiewicz versions
an upper bound of $\NEXP$ for strict satisfiability has been obtained by coalgebraic
methods~\cite{SchroderPattinson11} while we obtain \PSPACE upper
bounds for the non-expansive variants. The satisfiability problem in
fuzzy $\ALC$ with product semantics is decidable, but no complexity
bound has been given~\cite{CeramiEsteva22}.  For fuzzy description
logics over the G\"odel propositional base, the threshold
satisfiability problem remains in $\PSPACE$ for the basic logic
$\ALC$~\cite{CaicedoEA17}, and decidability is retained even in very
expressive logics \cite{BorgwardtEA16}. Similarly, over the Zadeh
base, reasoning is decidable even for highly expressive
logics~\cite{StoilosEA07,StoilosEA14}.  There is a tableau algorithm
for fuzzy $\ALC$ (over the empty TBox) that works with any continuous
t-norm~\cite{Baader2015}, subject to varying complexity. Our
coalgebraic algorithm is based on a quantitative extension of the
principle of reduction to \emph{one-step satisfiability},
i.e.~satisfiability in a small fragment of the logic that, roughly
speaking, prohibits nesting of modalities. This principle has
previously been used in the two-valued
setting~\cite{SchroderPattinson08,HausmannSchroder24} and, in a
different formulation from the one we employ here, in work on
\L{}ukasiewicz coalgebraic modal logic~\cite{SchroderPattinson11}.

\subparagraph*{Organization} We introduce our general framework of
non-expansive quantitative coalgebraic logic in
\Cref{nonexpfuzlog}. The technical development of our main result then
stretches across \Cref{onesteplogics,tableau,polspacebound}, where we
respectively introduce the key notion of one-step logic, a tableau
calculus, and a criterion for a complexity estimate of the tableau
algorithm. We use the logic of \emph{generally} as a running example;
additional examples (non-expansive fuzzy $\ALC$, quantitative fuzzy
$\ALC$, and quantitative metric modal logic) are discussed in
\Cref{sec:case-studies}.

\section{Non-Expansive Quantitative Coalgebraic Logic}\label{nonexpfuzlog}
We proceed to introduce our unifying framework of \emph{non-expansive quantitative
  coalgebraic logic}. This allows us to abstract from specific logics
and instead reason about whole classes of logics based on their
properties. We assume basic familiarity with category
theory~\cite{AdamekEA91}.

\subparagraph*{Syntax} We parametrize the logic over a set~$\Lambda$ of modal operators. For readability, we restrict the technical exposition to unary modal operators. The treatment of higher finite arities requires no more than additional indexing. Indeed, we will later sketch a treatment of propositional atoms as nullary modal operators. For $x\in[0,1]$ and $y\ge0$, write $x\oplus y=\min(1,x+y)$ and $x\ominus y=\max(0,x-y)$ for truncated addition and subtraction, respectively. The set $\mathcal{F}(\Lambda)$ of \emph{$\Lambda$-formulas} $\phi,\psi,\dots$ is defined by the grammar
\begin{equation*}
  \mathcal{F}(\Lambda)\owns\phi, \psi ::= 0  \mid \lnot \phi \mid \phi \ominus c \mid \phi \land \psi \mid \heartsuit \phi \qquad(c\in[0,1]\cap\Rat, \heartsuit\in\Lambda).
\end{equation*}
As usual, we define disjunction~$\lor$ by
$\phi\lor\psi=\neg(\neg\phi\land\neg\psi)$.  We write
$\mathsf{S}(\phi)$ for the set of subformulas of $\phi$, and
similarly~$\mathsf{S}(L)$ for the set of subformulas of formulas
in a set~$L$ of formulas.
Moreover, we write $\mathsf{S}_0(\phi)$ for the set of
\emph{propositional} subformulas of $\phi$, i.e.~subformulas of~$\phi$
not in the scope of a modal operator, and again write $\mathsf{S}_0(L)$ for the
set of all propositional subformulas of formulas in a set~$L$.

\subparagraph*{Semantics} We parametrize the semantics along several axes. First, we fix an endofunctor $F\colon \SET \rightarrow \SET$ on the category $\SET$ of sets and maps. We interpret the logic over \emph{$F$-coalgebras} or \emph{$F$-models}, i.e.~pairs $M=(X,\xi)$ where~$X$ is a set of \emph{states} and $\xi\colon X\to FX$ is a \emph{transition map} specifying for each state $x\in X$ a structured collection $\xi(x)$ of successors, with the notion of structure determined by~$F$. A basic example is the covariant powerset functor $F=\Pow$. In this case, $F$-coalgebras $\xi\colon X\to\Pow X$ assign to every state a \emph{set} of successors; an $F$-coalgebra is thus just a set~$X$ equipped with a binary successor relation, that is, a Kripke frame. Our main example will centrally involve the \emph{discrete distribution functor}~$\Dist$, which maps a set~$X$ to the set $\Dist X$ of discrete probability distributions on~$X$, given as maps $\mu\colon X\to[0,1]$ such that $\sum_{x\in X}\mu(x)=1$ (implying that $\mu(x)=0$ for all but countably many~$x$). For $A\subseteq X$, we write $\mu(A)=\sum_{x\in A}\mu(x)$; then, $\Dist$ sends a map $f\colon X\to Y$ to the map $\Dist f\colon\Dist X\to\Dist Y$ that takes image measures, i.e.~maps a distribution~$\mu$ on~$X$ to the distribution $\Dist f(\mu)$ on~$Y$ given by $\Dist f(\mu)(A)=\mu(f^{-1}[A])$. Coalgebras for~$\Dist$ thus associate to each state a discrete distribution over successor states, i.e.~they are (discrete-time) Markov chains, or \emph{probabilistic transition systems}.

A \emph{quantitative predicate lifting} for~$F$ is a natural transformation
\begin{equation*}
  \lambda\colon[0,1]^{(-)}\rightarrow [0,1]^{(-)}\circ F^{\op}
\end{equation*}
where $[0,1]^{(-)}\colon\SET^{\op}\to\SET$ is the \emph{contravariant
  fuzzy powerset functor}, which sends a set~$X$ to the set $[0,1]^X$
of maps $X\to[0,1]$, interpreted as $[0,1]$-valued
(\emph{quantitative}) subsets of~$X$ or \emph{predicates} on~$X$, and acts
on maps by precomposition; that is, for $f\colon X\to Y$ and
$B\in [0,1]^Y$, we have $[0,1]^f(B)=B\cdot f\in[0,1]^X$. Moreover,
$F^{\op}\colon\SET^{\op}\to\SET^{\op}$ is the opposite functor of~$F$,
which acts just like~$F$. Thus,~$\lambda$ is a family of
maps~$\lambda_X\colon [0,1]^X\rightarrow [0,1]^{FX}$, subject to the
\emph{naturality} condition saying that for $g\colon X\to Y$ and
$f\colon Y\to[0,1]$, we have
$\lambda_X(f\circ g)=\lambda_Y(f)\circ Fg$. We interpret each modal
operator $\heartsuit\in\Lambda$ as a predicate
lifting~$\llbracket\heartsuit\rrbracket$. The semantics assigns truth
values $\llbracket\phi\rrbracket_M(x)\in[0,1]$ to states $x\in X$ in
$F$-coalgebras $M=(X,\xi)$, defined recursively by
\begin{align*}
	\llbracket 0 \rrbracket_M (x) &= 0 &
	\llbracket \lnot \phi \rrbracket_M (x) &= 1 - \llbracket \phi \rrbracket_M (x) \\
	\llbracket \phi \ominus c \rrbracket_M (x) &= \llbracket\phi\rrbracket_M(x)\ominus c &
	\qquad \llbracket \phi \land \psi \rrbracket_M (x) &= \min (\llbracket \phi \rrbracket_M (x), \llbracket \psi \rrbracket_M (x)) \\
	\llbracket\heartsuit\phi\rrbracket_M(x) &= \llbracket\heartsuit\rrbracket(\llbracket{\phi}\rrbracket_M)(\xi(x)),
\end{align*}
exploiting that the truth values $\llbracket\phi\rrbracket_M(x)$
aggregate into a quantitative predicate
$\llbracket\phi\rrbracket_M\colon X\to[0,1]$, the
\emph{extension} of~$\phi$. Disjunction~$\lor$ is then defined by
$\phi\lor\psi=\neg(\neg\phi\land\neg\psi)$ as usual, so
$\llbracket \phi \lor \psi \rrbracket_M (x) = \max (\llbracket \phi
\rrbracket_M (x), \llbracket \psi \rrbracket_M (x))$.  We refer to the
entirety of the above data as the \emph{logic}, sometimes collectively
denoted~$\mathcal L$, and \emph{fix these data throughout the general
  technical development}.

\subparagraph*{Atoms} Before we go into the details of how our main example
of probabilistic logic fits into this
framework, we discuss the treatment of propositional atoms as nullary
modalities. Specifically, let~$\mathcal L$ be a logic given by data as
above, and let~$\mathsf{At}$ be a set of \emph{atoms}. Then we write
$\mathcal{L} + \mathsf{At}$ for the logic determined by the functor
$F_{\mathsf{At}}$ given on sets~$X$ by
$F_{\mathsf{At}} X= FX \times [0,1]^{\mathsf{At}}$ and the set
$\Lambda_{\mathsf{At}}=\Lambda \cup \{a \mid a \in {\mathsf{At}}\}$ of
modal operators, with interpretations of modal operators given by
applying the original interpretation of $\heartsuit\in\Lambda$ to the
first component~$t$ of a pair $(t,g)\in F_{\mathsf{At}}X$, and for
$a\in\mathsf{At}$ by $\llbracket a \rrbracket(f)(t,g)=g(a)$. Thus,
$\llbracket a \rrbracket$ ignores its argument~$f$, which we therefore
omit, effectively making~$a$ a nullary modality. We will show that
our complexity criterion is stable under passing from~$\mathcal{L}$ to
$\mathcal{L}+\mathsf{At}$, so most of the time we will just elide
propositional atoms in the presentation.

\subparagraph*{The logic of `generally'} Our running example will be the
modality `generally', originally developed in quantitative
probabilistic knowledge representation~\cite{SchroderPattinson11} (within a more expressive framework using the full \L{}ukasiewicz base). The
qualification \emph{probably} is understood as deeming the probability
of some property to be `high' in a vague sense, allowing that the
property itself may be vague.  One standard formal interpretation of
\emph{probably} is to take expected truth
values~\cite{Zadeh68,Hajek07,SchroderPattinson11}.  The following
alternative formal interpretation, with the suggested pronunciation
`\emph{generally}', has been proposed as having possibly better
computational properties~\cite{SchroderPattinson11} (`possibly'
because no tight lower bounds are known for `probably'):

Put $F=\Dist$, so the underlying models are probabilistic transition systems. Moreover, take $\Lambda=\{\GENERALLY\}$, with $\GENERALLY$ (\emph{generally}) interpreted by the predicate lifting
\begin{equation*}
	\llbracket \GENERALLY \rrbracket_X (f)( \mu) = \sup_{\alpha \in [0,1]} \{ \min (\alpha, \mu(\{x \in X \mid f (x) \geq \alpha\}) )\}
	\qquad\text{for $f \colon X \rightarrow [0,1]$, $\mu \in \mathcal{D}(X)$}.
\end{equation*}
For each candidate truth degree~$\alpha$, the lifting takes the
smaller of~$\alpha$ and the probability that the argument reaches that
degree. The supremum picks the best such threshold. The arising logic
is called the \emph{non-expansive logic of generally}, briefly
denoted~$\Lgen$. Recall that atoms are present implicitly, although
elided in the presentation. Without them, all states in
$\Dist$-coalgebras would be behaviourally equivalent, and hence
logically indistinguishable~\cite{WildSchroder22}.

The 
truth degree $\llbracket \GENERALLY \phi \rrbracket_M (x)\in[0,1]$ of
a formula~$\GENERALLY \phi$ at a state $x \in X$ in a $\Dist$-coalgebra $(X,\xi)$ is thus given by
\begin{equation*}
  \llbracket \GENERALLY \phi \rrbracket_M (x) = \sup_{\alpha \in [0,1]}  \min (\alpha, \xi(x)(\{y \in X \mid \llbracket \phi \rrbracket_M (y) \geq \alpha\}) )
\end{equation*}
(note that $\xi(x)$ is a probability distribution on~$X$).  If we
think of the transition function as expressing a probabilistic
relationship `associated with' between real-world entities, we can,
for instance, describe very professional football players who tend to
either play unfairly or suffer grave injuries by the formula
\begin{equation*}
  (\mathsf{professional}\ominus 0.2)\land\mathsf{football\_player}\land
  \GENERALLY\,(\mathsf{unfairness}\lor\mathsf{grave\_injury}).
\end{equation*}
Note how the emphasis `very professional' is reflected by the shift
${\ominus}\, 0.2$, which implies that an even higher degree of
professionality is required to give the formula a high truth
degree. Over the full \L{}ukasiewicz base (a more expressive
propositional base than in the above grammar), Schr\"oder and Pattinson give
upper bounds for strict satisfiability of \textsc{NExpTime} for their variant of
\emph{generally} and \textsc{ExpSpace} for \emph{probably}~\cite[Example~5.7]{SchroderPattinson11}.

As mentioned earlier, the logic of \emph{generally} closely relates to
Sugeno integrals \cite{1975TheoryOF} as used in fuzzy measure
theory. Let $(X, \Omega)$ be a measurable space (in our setting, this
space is implicitly discrete, i.e.~$\Omega=\Pow(X)$),
$f \colon X \rightarrow [0,1]$ an $\Omega$-measurable
function, and $g$ a monotone function $g\colon\Omega\to[0,1]$.
Then the \emph{Sugeno integral} of $f$ with respect to $g$ is
\begin{equation*}\textstyle
	\dashint_X f(x) \circ g = \sup_{\alpha \in [0,1]} (\min (\alpha, g(f_\alpha))),
\end{equation*}
where $f_\alpha := \{x \mid f(x) \geq \alpha\}$. This is exactly the
semantics of the `generally' modality when we take~$g$ to be the
successor probability measure of a state and~$f$ to be the evaluation
map of the argument formula.

\begin{remark}[L\'evy-Prokhorov distance]\label{rem:LP}
	The modality~$\GENERALLY$ does, much like the expectation modality \emph{probably} (which
	features both in the probabilistic
	$\mu$-calculus~\cite{MorganMcIver97,HuthKwiatkowska97} and in the
	characteristic modal logic of probabilistic transition
	systems~\cite{BreugelWorrell05}), equally relate to a view of
	probabilistic transition systems that regards
	elements of the underlying set as states of a concurrent system. In particular, as
	mentioned in the introduction, it follows from recent results that
	the behavioural distance on finitely branching probabilistic
	transition systems induced by the notion of
	\emph{$\epsilon$-bisimulation}~\cite{DesharnaisEA08} is
	characterized by  the logic of `generally' in the sense that behavioural
        distance coincides with logical distance. In more
        detail, one defines a notion of $\epsilon$-bisimulation on
        probabilistic systems such as labelled Markov chains in the
        same manner as standard notions of probabilistic bisimilarity,
        but allowing for a deviation of up to~$\epsilon$ between
        probabilities of sets of successors. This, in turn, gives rise
        to a notion of $\epsilon$-bisimilarity, and the induced
        \emph{$\epsilon$-bisimulation distance}, or just
        \emph{$\epsilon$-distance}, between states~$x,y$ is defined as
        the infimum over all~$\epsilon$ such that $x,y$ are
        $\epsilon$-bisimilar~\cite{DesharnaisEA08}. It has recently
        been shown~\cite{DesharnaisSokolova26} that
        $\epsilon$-distance coincides with a distance defined as a
        least fixpoint using the L\'evy-Prokhorov lifting of metrics
        to probability distributions, which is popular in statistics
        and machine learning due to its favourable stability
        properties~\cite{AolariteiEA25,BennounaEA23}. Specifically,
        given a metric space $(X,d)$, the associated L\'evy-Prokhorov
        distance $d_\mathsf{LP}$ on $\Dist X$ is given by
	\begin{equation*}
		d_\mathsf{LP}(\mu,\nu)=\inf\{\epsilon\ge 0 \mid\forall A\subseteq X
		.\,\nu(A^d_\epsilon)\ge\mu(A)-\epsilon\}
	\end{equation*}
	where we write $A^d_\epsilon=\{x\in X\mid d(x,A)\le\epsilon\}$ for
	the $\epsilon$-neigh\-bour\-hood of~$A$ under~$d$ (and, as usual,
	$d(x,A)=\inf_{y\in A}d(x,y)$).  
	Additionally, it has been proved recently that the L\'evy-Prokhorov distance~$d_\mathsf{LP}$ is induced by
	\emph{generally} in the sense of a generalized categorical
	\emph{Kantorovich lifting}~\cite{BaldanEA18,WildSchroder22} of
	metrics along functors.
	This results in the representation
	\begin{equation*}
		d_\mathsf{LP}(\mu,\nu)=\sup_{f\colon (X,d)\to_1 [0,1]}\GENERALLY(f)(\nu)-\GENERALLY(f)(\mu),
	\end{equation*}
	where $f\colon (X,d)\to_1[0,1]$ denotes that $f$ ranges over the non-expansive functions from $(X,d)$ to $[0,1]$ (equipped with Euclidean distance).

	By general (coalgebraic) quantitative Hennessy-Milner
        theorems~\cite{KonigMikaMichalski18,WildSchroder22,ForsterEA26},
        this implies that the modal logic of \emph{generally}
        characterizes $\epsilon$-distance in the sense that the
        induced logical distance on finitely branching probabilistic
        transition systems is precisely $\epsilon$-distance.
\end{remark}

\medskip

\subparagraph*{Tableau sequents} We introduce some key technical
notions regarding the labels of nodes in tableaux. We keep the notion
of tableau sequent, defined next, general enough to serve as the
syntactic core of both the one-step logic (\Cref{onesteplogics}) and
the tableau method for the full modal logic (\Cref{tableau}).

\begin{convention}
  Throughout, we let the placeholder~$\llparenthesis$ stand for opening brackets $($ or $[$, and~$\rrparenthesis$ for closing brackets $)$ or $]$. For a nonempty interval~$I\subseteq[0,1]$, its \emph{endpoints} are $\inf I$ and $\sup I$, denoted by~$a$ and~$b$ when we write $I=\llparenthesis a,b\rrparenthesis$. The empty interval has no endpoints.
\end{convention}
Let $L$ be a set of \emph{labels}. A \emph{tableau sequent over~$L$}
is a finite set of \emph{(tableau) literals}, which are expressions of
the form $\ell \tin I$ where $\ell \in L$ and $I \subseteq [0,1]$ is
an interval (possibly empty) with rational endpoints. The endpoints of a
literal are those of its interval, and the endpoints of a sequent are
those of its literals (possibly more than two). A \emph{tableau sequent} (with~$L$
unspecified) is a tableau sequent over the set $\mathcal{F}(\Lambda)$
of formulas. A tableau sequent~$\Gamma$ over~$L$ is
\emph{clean/complete/exact} if for each $\ell \in L$,~$\Gamma$
contains at most/at least/exactly one literal of the shape
$\ell \tin I$.  For $\ell\in L$, we write $\Gamma(\ell)$ for the
intersection of all~$I$ such that $(\ell\tin I)\in\Gamma$. Of course,
if~$\Gamma$ is clean, then~$\Gamma(\ell)$ is just the unique
interval~$I$ such that $(\ell\tin I)\in\Gamma$, or $[0,1]$ if no such
interval exists. We apply a \emph{substitution}
$\sigma\colon L \rightarrow L'$ to~$\Gamma$ by taking $\Gamma\sigma$
to be the exact tableau sequent over $\sigma[L]$ given by
$(\Gamma\sigma) (\ell') := \bigcap_{\ell \in L, \sigma (\ell) =
  \ell'}\Gamma(\ell)$ for $\ell'\in\sigma[L]$.

When~$L$ is a set of formulas, a state $x$ in a $F$-coalgebra $M$
\emph{satisfies} a tableau literal $\phi \tin I$ (notation
$M, x \models \phi \tin I$) if
$\llbracket \phi \rrbracket_M (x) \in I$, and a tableau sequent
$\Gamma$ if $M, x \models \phi \tin I$ for every literal
$(\phi \tin I) \in \Gamma$ (notation $M, x \models \Gamma$). If $M,x$
satisfying~$\Gamma$ exist, then $\Gamma$ is \emph{satisfiable}.
\begin{remark}
  Usually, we care more about \emph{validity} than
  satisfiability. That is, we want to prove that a given tableau
  literal is satisfied in all states of all models. As usual, validity
  can be reduced to satisfiability; for instance, to check that
  $\phi \tin ( a,b]$ is valid, check that both $\phi \tin [0,a]$ and
  $\phi \tin ( b, 1]$ are unsatisfiable.
\end{remark}
We write $\operatorname{bits}(a)$ for the binary length of a nonnegative integer~$a$, counting~$0$ as one bit, and take the size of a rational constant $c=a/b$ in lowest terms to be $1+\operatorname{bits}(a)+\operatorname{bits}(b)$. The syntactic sizes of formulas, literals and sequents are defined recursively by
\begin{gather*}
  \lvert 0 \rvert = 1 \qquad \lvert \lnot \psi \rvert = \lvert \psi \rvert + 1 \qquad \lvert \psi_1 \land \psi_2 \rvert = \lvert \psi_1 \rvert + \lvert \psi_2 \rvert + 1 \qquad \lvert \psi \ominus c \rvert = \lvert \psi \rvert + \lvert c \rvert + 1\\
  \lvert \heartsuit \psi \rvert = \lvert \psi \rvert + \lvert \heartsuit \rvert \qquad \lvert \phi \tin \llparenthesis a, b \rrparenthesis \rvert = \lvert \phi \rvert + \lvert a \rvert + \lvert b \rvert + 3 \qquad \lvert \Gamma \rvert = \textstyle\sum_{l \in \Gamma} \lvert l \rvert
\end{gather*}
where $\psi,\psi_1,\psi_2$ are formulas, $\phi\tin\llparenthesis a,b\rrparenthesis$ is a literal, and~$\Gamma$ is a sequent. The~$3$ in the size of a literal counts one unit for each of the two brackets and one unit for~$\tin$. For an empty interval, we put $\lvert\phi\tin\emptyset\rvert=\lvert\phi\rvert+2$. Unless otherwise noted, we put $\lvert\heartsuit\rvert=1$.

\section{The One-Step Logic}\label{onesteplogics}
A core theme of coalgebraic logic is to reduce properties of the full
modal logic to properties of a much simpler one-step logic, in which
dedicated \emph{variables}, i.e.~placeholders for formulas, each
appear under exactly one modality (in particular, the one-step logic
precludes nesting of modalities). We next introduce the relevant
notion of one-step logic for our present setup; we will later reduce
the satisfiability problem of the full logic to calls of
satisfiability checking in the one-step logic. Recall that we have
globally fixed a logic~$\mathcal{L}$ with endofunctor~$F$ and modal
operators~$\Lambda$.

\begin{defn}
  \begin{enumerate}
  \item For a set $V$ of \emph{variables}, write $\Lambda(V) = \{\heartsuit v \mid v \in V, \heartsuit \in \Lambda\}$.
  \item The set $\operatorname{Prop}(\Lambda(V))$ of \emph{one-step
      formulas} over~$V$ is defined by the grammar
    \begin{equation*}
      \operatorname{Prop}(\Lambda(V))\owns\phi, \psi ::= 0 \mid \lnot \phi \mid \phi \ominus c \mid \phi \land \psi \mid \heartsuit v
    \end{equation*}
    where $v \in V$ is a variable, $c$ is a constant, and $\heartsuit \in \Lambda$.
  \item A \emph{one-step $F$-model} (or just \emph{one-step model}) over~$V$ is a tuple $M=(X,\tau,t)$ consisting of a nonempty set~$X$, an element $t\in FX$, and a \emph{valuation} $\tau\colon V\to[0,1]^X$.
  An empty carrier could be replaced by a singleton~$1$ with an arbitrary valuation, since for the unique map~$i\colon\emptyset\to 1$, passing from~$t\in F\emptyset$ to~$Fi(t)\in F1$ preserves all one-step formula values by naturality.
  \item The \emph{one-step extension}
    $\llbracket - \rrbracket_M\colon \operatorname{Prop}(\Lambda(V))
    \rightarrow [0,1]$ of one-step formulas over a one-step model~$M$
    is defined by
    \begin{gather*}
      \llbracket 0 \rrbracket_M = 0 \qquad \llbracket \lnot \phi \rrbracket_M = 1 - \llbracket \phi \rrbracket_M \qquad \llbracket \phi \ominus c \rrbracket_M = \llbracket\phi\rrbracket_M\ominus c\\
      \llbracket \phi \land \psi \rrbracket_M = \min (\llbracket \phi \rrbracket_M, \llbracket \psi \rrbracket_M) \qquad \llbracket \heartsuit v \rrbracket_M = \llbracket \heartsuit \rrbracket_X (\tau(v)) (t)
    \end{gather*}
  \item A \emph{one-step tableau sequent} is a tableau sequent
    $\Gamma$ over a set $L \subseteq \operatorname{Prop}(\Lambda(V))$
    of one-step formulas. A one-step $F$-model~$M$
    \emph{satisfies}~$\Gamma$ (notation $M \models \Gamma$) if
    $\llbracket \phi \rrbracket_M \in I$ for each literal
    $(\phi \tin I) \in \Gamma$. If there exists~$M$ such that
    $M\models\Gamma$, then~$\Gamma$ is \emph{satisfiable}.
  \end{enumerate}\smallskip
  
  \noindent Syntactic size is defined analogously to the full logic,
  with variables counted as having size~$1$.
\end{defn}

\begin{defn}
  Let~$\Gamma$ be a tableau sequent. First form the exact sequent~$\Gamma^{\mathrm{ex}}$ over the formulas occurring in~$\Gamma$ by intersecting the intervals of all literals mentioning the same formula. The \emph{top-level decomposition} $(V,\sigma_\Gamma,\Gamma^\sharp)$ is obtained from~$\Gamma^{\mathrm{ex}}$ by replacing the argument of each occurrence of a modality that is not in scope of another modality by a distinct fresh variable; then,~$V$ is the set of these fresh variables. The substitution $\sigma_\Gamma\colon V\to\mathcal{F}(\Lambda)$ maps each variable back to the formula that it replaced; substitutions  extend to formulas in the usual way. In particular, $\Gamma^\sharp\sigma_\Gamma=\Gamma^{\mathrm{ex}}$.
\end{defn}
It is clear that the top-level decomposition is unique up to renaming
the fresh variables. The following simple fact is the driving
principle behind reductions to the one-step logic:

\begin{lemma}\label{localreduction}
  Let $(V,\sigma_\Gamma,\Gamma^\sharp)$ be the top-level decomposition of a tableau sequent~$\Gamma$. Then $\Gamma$ is satisfiable iff there is a one-step model $(X,\tau,t)$ satisfying~$\Gamma^\sharp$ such that for every $x\in X$, there is a state $y_x$ in a model $M_x$ such that $\llbracket\sigma_\Gamma(v)\rrbracket_{M_x}(y_x)=\tau(v)(x)$ for all~$v\in V$.
\end{lemma}
\begin{proof}[Proof (sketch)]
  From a model satisfying~$\Gamma$, one may take the transition structure at the satisfying state and evaluate the modal arguments at all states of the model. Conversely, one can take disjoint copies of models jointly realizing the prescribed successor values and add a fresh root with the given one-step transition structure.
\end{proof}

\begin{remark}\label{rem:fmp}
  The one-step finite model property for exact truth values requires that the values of any finite family
  of modal formulas in a one-step model can be reproduced on a finite carrier with successor valuations
  in the same prescribed set. By a local reduction argument  much as in~\Cref{localreduction}~\cite[Definitions~4.3--4.4 and Theorem~4.7]{SchroderPattinson11},
  this implies the finite model property for mixed sequents. In the proof, one takes as the prescribed set of valuations the jointly realizable tuples of values of the modal arguments. By induction on modal depth, we can then build a finite model by attaching a finite model to each state in a finite one-step model.
\end{remark}
\section{A Tableau Calculus Based on Modal Co-Rules}\label{tableau}
We next introduce a tableau calculus to decide satisfiability of
tableau sequents.  The calculus has a generic propositional part and
is parametrized by a set of logic-specific modal \emph{co-rules}
(\Cref{def:co-rules}), which work in a dual fashion to standard
(modal) tableau rules: Both rules and co-rules consist of a premise
and possibly several conclusions. In calculi based on modal rules, a
tableau sequent is deemed satisfiable if all matching rules have a
satisfiable conclusion; contrastingly, in our calculus, a
propositionally saturated tableau sequent is deemed satisfiable if
there exists a matching modal co-rule all whose conclusions are
satisfiable (while the propositional part of the calculus is still
based on rules applied in the standard way). This is reminiscent of
the approach used in two-valued algorithms based on one-step
satisfiability checking \cite{Kupke2022}. Co-rules come with notions
of soundness and completeness formulated at the level of the one-step
logic; soundness and completeness of the calculus are thus based
centrally on \Cref{localreduction}.

We start by introducing the generic propositional part of the tableau
calculus, which we apply both to one-step tableau sequents and to
tableau sequents in the full logic. The \emph{propositional rules} are
shown in \Cref{tableau:prop}. In the rules, we use the notation
\begin{equation*}
  1-I := \{1-x \mid x \in I\} \qquad
  I+c := \{x+c \mid x\in I\}
\end{equation*}
for intervals $I \subseteq [0,1]$. As usual, we omit set brackets
around literals and denote set union as a comma, so
$\Gamma,\phi\tin I$ denotes $\Gamma\cup\{\phi\tin I\}$. As indicated
above, the rules each consist of a premise and possibly several
conclusions, and are read as saying that the premise is deemed
satisfiable if one of the conclusions is satisfiable. The rules
explicitly include a \emph{context}~$\Gamma$; a rule \emph{applies to}
or \emph{matches} a tableau sequent if the tableau sequent is exactly
the rule premise (including the context). A tableau sequent is
\emph{propositionally saturated} if no propositional rule applies to
it. We formalize the construction of propositional tableau branches as
follows:
\begin{defn}
  Let~$\Gamma$ be a (one-step) tableau sequent. A \emph{propositional
    tableau (branch)} for~$\Gamma$ is a list of labelled nodes
  $x_1,\ldots,x_n$ such that $x_1$ has label~$\Gamma$, any consecutive
  nodes have as labels the premise and a conclusion, respectively, of
  a rule of~\Cref{tableau:prop}, and the label of~$x_n$ is
  propositionally saturated. The tableau~$G$ is \emph{open} if the
  label~$Y$ of~$x_n$ is not~$\bot$. We then write
  $\Gamma_G:=\{(\phi\tin I)\in Y\mid\phi\ne0\}$.  In an open branch,
  every literal of the form $0\tin I$ in the last label has $0\in I$
  by rule~$(\text{Ax 0})$ of~\Cref{tableau:prop}. We assume silently
  that these literals are dropped; this preserves satisfiability, and
  leaves only \emph{modal literals}, i.e.~literals of the form
  $\heartsuit\phi\tin I$, in the last label.
\end{defn}
\begin{table*}[t]
	\centering
	\caption{Propositional tableau rules}
	\label{tableau:prop}
	\begin{tabular}{c}
		\toprule
		$\displaystyle
		(\text{Ax}) \  \frac{\Gamma, \phi \tin \emptyset}{\bot} \qquad (\text{Ax 0}) \  \frac{\Gamma, 0 \tin I}{\bot} \quad (\text{if } 0 \notin I) \qquad (\cap) \frac{\Gamma, \phi \tin I, \phi \tin J}{\Gamma, \phi \tin I \cap J}
		$
		\\[1em]
		$\displaystyle
		(\lnot) \  \frac{\Gamma, \lnot \phi \tin I}{\Gamma, \phi \tin (1-I)} \qquad (\ominus) \   \frac{\Gamma, \phi \ominus c \tin I}{\Gamma, \phi \tin ((I+c)\cap[0,1])} \quad (\text{if } 0 \notin I) \qquad (\ominus') \   \frac{\Gamma, \phi \ominus c \tin [0 ,b \rrparenthesis}{\Gamma, \phi \tin [ 0, b +  c \rrparenthesis \cap [0,1]}
		$
		\\[1em]
		$\displaystyle
		(\land) \   \frac{\Gamma, \phi \land \psi \tin \llparenthesis a,b \rrparenthesis}{\Gamma, \phi \tin \llparenthesis a,b \rrparenthesis, \psi \tin \llparenthesis a,1 ] \quad \Gamma, \phi \tin \llparenthesis a,1 ], \psi \tin \llparenthesis a,b \rrparenthesis}
		$
		\\[1em]
		\bottomrule
	\end{tabular}
\end{table*}
\noindent Soundness and completeness of the propositional calculus are
stated as follows:
\begin{lemma}\label{tableau:correctness}
  A (one-step) tableau sequent~$\Gamma$ is satisfiable iff there
  exists an open propositional tableau $G$ for $\Gamma$ such
  that~$\Gamma_G$ is satisfiable.
\end{lemma}
We note in passing that the propositional calculus witnesses
decidability in \NP of purely propositional reasoning over the
non-expansive propositional base.

We proceed to introduce the modal part of the tableau calculus:

\begin{defn}
  Let $\tau\colon V\to[0,1]^X$ be a valuation on a nonempty set~$X$. The pair $(X,\tau)$ (or just~$\tau$), and any one-step model $(X,\tau,t)$ extending it, \emph{realizes} a tableau sequent~$\Delta$ over $V$ if there is $x\in X$ such that $\tau(\ell)(x)\in I$ for every $(\ell\tin I)\in\Delta$.
\end{defn}

\begin{defn}\label{def:co-rules}
  \begin{enumerate}
  \item A \emph{modal co-rule} for~$\Gamma$ has the form $\frac{\Gamma}{\Delta_1\ \mid\ \ldots\ \mid\ \Delta_n}$ where $n\ge0$, the \emph{premise}~$\Gamma$ is a propositionally saturated one-step tableau sequent, and each \emph{conclusion}~$\Delta_i$ is exact over the set~$V_\Gamma$ of variables occurring in~$\Gamma$. We say that the co-rule \emph{matches} or \emph{applies to~$\Gamma$}.
  \item A modal co-rule
    $\frac{\Gamma}{\Delta_1\ \mid\ \ldots\ \mid\ \Delta_n}$ is
    \emph{one-step sound} if every valuation
    $\tau\colon V_\Gamma\to [0,1]^X$ that realizes all~$\Delta_i$
    extends to a one-step model $(X,\tau,t)$ that
    satisfies~$\Gamma$. 
  \item A set $\mathcal{R}$ of modal co-rules is \emph{one-step
      tableau complete} if whenever~$\Gamma$ is satisfied in a
    one-step model $M=(X, \tau, t)$, then there exists a modal co-rule
    $\frac{\Gamma}{\Delta_1 \ \mid \ \ldots \ \mid \ \Delta_n}$ in
    $\mathcal{R}$ such that all $\Delta_i$ are realized in $M$.
  \item For a premise~$\Gamma$ and a set~$\mathcal{R}$ of modal co-rules, write $\mathcal{R}_\Gamma=\{R\in\mathcal{R}\mid R\text{ matches }\Gamma\}$ for the set of co-rules in~$\mathcal{R}$ matching~$\Gamma$.
	\end{enumerate}
\end{defn}

\noindent As indicated above, a co-rule
$\frac{\Gamma}{\Delta_1 \ \mid \ \ldots \ \mid \ \Delta_n}$ is read as
saying that~$\Gamma$ is satisfiable if all~$\Delta_i$ are
satisfiable. Intuitively, we can describe a one-step sound and tableau
complete set of modal co-rules as a representation of all possible
combinations of successor states and their values for each variable
that can be extended to a one-step model satisfying $\Gamma$.

We introduce assumptions on the computability of sets of co-rules:
\begin{defn}\label{definition:rulescheme2}
  A \emph{modal co-rule scheme} is a one-step tableau complete set~$\mathcal{R}$ of one-step sound co-rules such that each set of matching co-rules~$\mathcal{R}_\Gamma$ is finite and equipped with an enumeration $[1,\ldots,k]\rightarrow\mathcal{R}_\Gamma$, and there are total computable functions $\matches_{\mathcal{R}}$, $\concls_{\mathcal{R}}$, and $\concl_{\mathcal{R}}$ such that
    \begin{enumerate}
    \item $\matches_{\mathcal{R}}(\Gamma)$ is an upper bound on $\lvert\mathcal{R}_\Gamma\rvert$,
    \item $\concls_{\mathcal{R}}(\Gamma,i)$ is an upper bound on the
      number of conclusions of the $i$-th co-rule
      in~$\mathcal{R}_\Gamma$, or~$\bot$ if~$|\mathcal R_\Gamma|<i$;
    \item $\concl_{\mathcal{R}}(\Gamma,i,j)$ returns the
      $j$-th conclusion of the $i$-th co-rule in $\mathcal R_\Gamma$,
      or $\bot$ if either the $i$-th co-rule or its $j$-th conclusion
      does not exist.
    \end{enumerate}
  \end{defn}

\begin{remark}
  The simplistic enumeration of matching co-rules and their
  conclusions in \Cref{definition:rulescheme2} circumvents more
  sophisticated representations in terms of (co-)rule codes as used in
  previous work in two-valued coalgebraic modal
  logic~\cite{SchroderPattinson09,CirsteaEA11}, which we avoid in the
  interest of readability.

\end{remark}

If a logic admits a modal co-rule scheme, we can decide satisfiability using~\Cref{alg:sat}. We assume $F1\ne\emptyset$, where~$1$ is a singleton, since otherwise no nonempty $F$-coalgebra can exist by functoriality. First, we use top-level decomposition to reduce to the one-step logic, then eliminate propositional operators using the propositional tableau calculus. An empty open leaf can be accepted immediately. Otherwise, we guess a modal co-rule, resubstitute the argument formulas for the variables, and recursively check each conclusion for satisfiability. A run that reaches $\mathbf{accept}$ witnesses satisfiability, while $\mathbf{abort}$ marks a non-accepting run.

\SetKwProg{Fn}{Function}{:}{}
\SetKwFunction{SatR}{\ensuremath{\operatorname{Sat}_{\mathcal R}}}
\SetKw{Abort}{abort}
\SetKw{Accept}{accept}

\begin{algorithm*}[!ht]
	\caption{\normalsize Satisfiability procedure for a modal co-rule scheme $\mathcal{R}$}
	\label{alg:sat}
	\Fn{\textnormal{\SatR{$\Gamma$}}} {
		\KwIn{a tableau sequent $\Gamma$}
		\KwOut{accepts iff $\Gamma$ is satisfiable}
		construct a top-level decomposition $(V, \sigma_\Gamma, \Gamma^\sharp)$ for $\Gamma$\\
		(non-deterministically) construct a propositional tableau $G$ for $\Gamma^\sharp$\\
		\If{$G$ is open} {
			\If{$\Gamma^\sharp_G = \emptyset$} {
				\Accept
			}
			(non-deterministically) choose an index for a co-rule $1 \leq i \leq \matches_{\mathcal{R}}(\Gamma^\sharp_G)$\\
			\If{$\concls_{\mathcal{R}}(\Gamma^\sharp_G, i) \neq \bot$} {
				\ForAll{$1 \leq j \leq \concls_{\mathcal{R}}(\Gamma^\sharp_G, i)$} {
					\If{$\concl_{\mathcal{R}}(\Gamma^\sharp_G, i, j) \neq \bot$} {
						\If{\SatR{$\concl_{\mathcal R}(\Gamma^\sharp_G,i,j)\sigma_\Gamma$} is not accepting}{
							\Abort
						}
					}
				}
				\Accept
			}
		}
	}
\end{algorithm*}

We can view the algorithm as constructing a tableau in the full
logic, where the (non-deterministic) choices correspond to
OR-branching in the tableau, while checking all the conclusions of a
modal co-rule corresponds to AND-branching.

\begin{theorem}\label{thm:mainCorrect}
  Given a modal co-rule scheme for~$\mathcal{L}$ and $F1\ne\emptyset$, \Cref{alg:sat} accepts a tableau sequent~$\Gamma$ iff~$\Gamma$ is satisfiable.
\end{theorem}

\section{Polynomially Space-Bounded Logics}\label{polspacebound}
Finally, we introduce conditions on a modal co-rule scheme under which
we establish $\PSPACE$ decidability of satisfiability, by showing that
\Cref{alg:sat} uses only polynomial space. We continue to use the logic of
generally, $\Lgen$, as a running example to illustrate how the
relevant conditions are established.  We note at the same time
that~$\Lgen$ is a comparatively involved case; the additional examples
we discuss in \Cref{sec:case-studies} are easier, witnessing that the
generic part of the development absorbs most of the technical
difficulty. We will show that the requisite conditions are preserved
under extension of a logic with atoms as per \Cref{nonexpfuzlog}, and
thus elide the explicit treatment of atoms in the examples.

\begin{samepage}
\begin{defn}\label{def:exp-branching}
  A logic~$\mathcal{L}$ is \emph{exponentially branching} if it admits
  a modal co-rule scheme~$\mathcal{R}$ and an exponential function~$f_{\mathcal L}$
  such that $\matches_{\mathcal{R}}(\Gamma)\le f_{\mathcal L}(\lvert\Gamma\rvert)$
  for every propositionally saturated one-step tableau sequent~$\Gamma$
  and $\concls_{\mathcal{R}}(\Gamma,i)\le f_{\mathcal L}(\lvert\Gamma\rvert)$ for
  each matching co-rule~$i$ (using the notation of~\Cref{definition:rulescheme2}).
\end{defn}
\end{samepage}
When we, from here on, refer to the modal co-rule scheme in an exponentially branching logic $\mathcal{L}$, we mean any fixed modal co-rule scheme as per \Cref{def:exp-branching}.

The proof of exponential branching for~$\Lgen$ will rely on Carathéodory's theorem, which states that any point in the convex hull of a set in~$\mathbb{R}^{d}$ can be written as a convex combination of at most~$d+1$ elements. This will allow us to restrict to co-rules with linearly many conclusions in the number of modal literals in the premise.

The next lemma reduces a bound on $\GENERALLY v$ to a bound on the probability of a set of successors. For a fixed threshold, a successor contributes its probability exactly when its truth value passes the corresponding test. We write $(1-f)(x)=1-f(x)$ and $f_{\triangleright a}=\{x\in X\mid f(x)\triangleright a\}$ for $f\in [0,1]^X$, $\triangleright\in\{{>},{\ge}\}$, and $a\in[0,1]$.
\begin{lemma}\label{lemma:generallybounds}
  For every one-step model $M = (X,\tau,t)$ over~$V$ in
  $\mathcal{L}_{\mathsf{gen}}$ (so $t\in\Dist X$), $v\in V$,
  $a,b\in[0,1]$, and~$\triangleright,\triangleright'\in\{{>},{\ge}\}$, we have
  \begin{equation*}
    \llbracket \GENERALLY v \rrbracket_M \triangleright a
    \iff t(\tau(v)_{\triangleright a}) \triangleright a \quad
    \text{and}\quad
    b \triangleright'\llbracket \GENERALLY v \rrbracket_M 
    \iff t((1-\tau(v))_{\triangleright' 1-b}) \triangleright' 1-b
  \end{equation*}
\end{lemma}

\begin{lemma}\label{theorem:generallyexpbounded}
  If a one-step model in~$\Lgen$ satisfies a propositionally saturated
  one-step tableau sequent with~$n$ modal literals, then the sequent
  is satisfiable on at most~$2n+1$ states of that model with the
  valuation restricted to those states.
\end{lemma}
\begin{proof}[Proof (sketch)]
  The idea is to represent a one-step model by the probabilities used in~\Cref{lemma:generallybounds}. For a propositionally saturated one-step sequent~$\Gamma$ with~$n$ modal literals, let $u\in[0,1]^{2n}$ record the two probabilities $t(\tau(v)_{\triangleright a})$ and $t((1-\tau(v))_{\triangleright' 1-b})$ for each interval $\Gamma(\GENERALLY v)=\llparenthesis a,b\rrparenthesis$, with the comparisons chosen according to the brackets.

  Each successor contributes a Boolean vector in~$\{0,1\}^{2n}$ whose entries indicate membership in those threshold sets. The probability vector~$u$ is a convex combination of the Boolean vectors that actually occur: the weight of each vector is the total probability of successors with that vector. By Carathéodory's theorem, at most $2n+1$ of these vectors suffice. Choose one representative successor for each selected vector and give it the corresponding weight. The resulting model has the same probabilities~$u$ and therefore satisfies the same tested bounds by~\Cref{lemma:generallybounds}.
\end{proof}

\begin{example}\label{example:generally-corule1}
		We illustrate the construction from \Cref{theorem:generallyexpbounded} on the tableau sequent $\Gamma := \{\lnot \GENERALLY p \tin [0.5,0.5],\ \GENERALLY \lnot p \tin [0,0.5)\}$.
		After top-level decomposition and elimination of the outer negation, we obtain $\Delta := \{\GENERALLY v_1 \tin [0.5,0.5],\ \GENERALLY v_2 \tin [0,0.5)\}$. The following six sequents describe the possible argument values at a single successor. \Cref{example:generally-corule2} uses them as candidate conclusions for the co-rules with premise~$\Delta$.
		\begin{align*}
			&q_1 := \{v_1 \tin [0.5,0.5], v_2 \tin [0,0.5)\},
			&&q_2 := \{v_1 \tin [0.5,0.5], v_2 \tin [0.5,1]\},\\
			&q_3 := \{v_1 \tin [0,0.5), v_2 \tin [0,0.5)\},
			&&q_4 := \{v_1 \tin [0,0.5), v_2 \tin [0.5,1]\},\\
			&q_5 := \{v_1 \tin (0.5,1], v_2 \tin [0,0.5)\},
			&&q_6 := \{v_1 \tin (0.5,1], v_2 \tin [0.5,1]\}.
		\end{align*}
		Their vector representations, indicating membership in the sets $\tau(v_1)_{\ge 0.5}$, $(1-\tau(v_1))_{\ge 0.5}$, $\tau(v_2)_{\ge 0}$ and $(1-\tau(v_2))_{>0.5}$, are
		\begin{align*}
			&q_1 \cong (1,1,1,1)^t,
			&&q_2 \cong (1,1,1,0)^t,
			&&q_3 \cong (0,1,1,1)^t,\\
			&q_4 \cong (0,1,1,0)^t,
			&&q_5 \cong (1,0,1,1)^t,
			&&q_6 \cong (1,0,1,0)^t.
		\end{align*}
\end{example}

\begin{example}\label{example:notEverythingExpBounded}
  Consider the powerset functor~$\Pow$ and a
  modality~$\heartsuit$ for~$\Pow$ interpreted by the predicate lifting
  $\llbracket\heartsuit\rrbracket$ given by
  $\llbracket\heartsuit\rrbracket_X(f)(U)=\sup_{x\in U,\
    f(x)\ne1}f(x)$. The sequent $\Gamma=\{\heartsuit v\tin[1,1]\}$ is
  satisfiable but every one-step model satisfying it is
  infinite. Hence this logic fails exponential branching.
  Indeed, one-step tableau completeness would give a matching co-rule
  whose conclusions are realized in the model. Retain one witness per
  conclusion, or one state if there are no conclusions. By one-step
  soundness, the restricted valuation would extend to a finite model
  satisfying~$\Gamma$.
\end{example}

\begin{lemma}\label{theorem:generallyexpbranching}
  The non-expansive logic of `generally',
  $\mathcal{L}_{\mathsf{gen}}$, is exponentially branching.
\end{lemma}
\begin{proof}[Proof (sketch)]
  Following the proof of \Cref{theorem:generallyexpbounded}, we need
  only $2n+1$ successor states, which we can represent as elements of
  $\{0,1\}^{2n}$ (where $n$ is the number of modal formulas
  in~$\Gamma$) that indicate whether a single-state one-step model
  would satisfy a specific bound for the formulas $\GENERALLY v$ or
  not. In the full model, this indicates whether a state helps satisfy
  a specific bound or counts against that bound. In total, this means
  there are at most $2^{2n \cdot (2n+1)}$ possible sets of successor
  states (a generous upper bound on the number of subsets with at most
  $2n+1$ elements of a $2^{2n}$-element set) which we need to
  consider.  Some of these successor sets may not admit a suitable
  probability distribution whose vector representation actually
  satisfies the inequalities imposed by $\Gamma$, so we filter those
  out.  Each remaining successor set gives rise to one modal co-rule,
  where each of the up to $2n+1$ successor states corresponds to an
  exact tableau sequent over $V$, where the entries $0$ and $1$
  determine the interval within which each $v$ is allowed to range for
  this particular state; these sequents form the conclusions of the
  co-rule.  We therefore obtain an at most exponential number of modal
  co-rules, which are one-step sound and one-step tableau complete for
  $\Gamma$ by construction. This procedure, then, defines our modal
  co-rule scheme, with
  $\matches_{\mathcal{R}}(\Gamma) = 2^{2n \cdot (2n+1)}$ and
  $\concls_{\mathcal{R}}(\Gamma, i) = 2n+1$ if the $i$-th co-rule exists and $\bot$ otherwise.
\end{proof}

\begin{example}\label{example:generally-corule2}
  Continuing~\Cref{example:generally-corule1}, we show that~$\Gamma$ is unsatisfiable by constructing the modal co-rules for the premise~$\Delta$. Let~$u_i$ be the Boolean vector for~$q_i$ in that example. Take a nonempty subset $S\subseteq\{1,\ldots,6\}$ of size at most~$5$. It gives a co-rule with premise~$\Delta$ and conclusions~$q_i$ for $i\in S$ if we can choose weights $p_i\ge0$ summing to~$1$, with $p_i=0$ for $i\notin S$, such that $(w_1,w_2,w_3,w_4)^t=\sum_{i=1}^6 p_i u_i$ satisfies $w_1\ge0.5$, $w_2\ge0.5$, $w_3\ge0$, and $w_4>0.5$.

  We can make do with just two co-rules. If $p_1=0$, then $w_4>0.5$ gives $p_3+p_5>0.5$, whereas $w_1,w_2\ge0.5$ imply $p_3\le0.5$ and $p_5\le0.5$, respectively. Thus every feasible set of conclusions contains either~$q_1$ or both~$q_3$ and~$q_5$. Both sets are feasible, by putting all the weight on~$q_1$ or splitting it equally between~$q_3$ and~$q_5$. Hence a sound and one-step tableau complete set of co-rules for~$\Delta$ is
  \begin{equation*}
    \frac{\Delta}{q_1}
    \qquad
    \frac{\Delta}{q_3 \mid q_5}.
  \end{equation*}
  After resubstituting~$p$ for~$v_1$ and~$\neg p$ for~$v_2$, the second literal in both~$q_1$ and~$q_3$ requires the value of~$p$ to exceed~$0.5$. Their first literals require it to equal~$0.5$ and to be strictly below~$0.5$, respectively. Neither co-rule has all its conclusions satisfiable after resubstitution, so~$\Gamma$ is unsatisfiable.
\end{example}

\begin{defn}
  A logic~$\mathcal{L}$ is \emph{polynomial-space bounded} if~$\mathcal{L}$ is exponentially branching and the associated modal co-rule scheme~$\mathcal{R}$ can be chosen such that for $i,j\in\mathbb{N}$ and a propositionally saturated one-step tableau sequent~$\Gamma$, we can compute $\matches_{\mathcal{R}}(\Gamma)$, $\concls_{\mathcal{R}}(\Gamma,i)$, and $\concl_{\mathcal{R}}(\Gamma,i,j)$ in polynomial space in $\lvert\Gamma\rvert$.
  Moreover, for each co-rule in~$\mathcal{R}$, every endpoint of a conclusion belongs to the additive subgroup of~$\Rat$ generated by~$1$ and the endpoints of the premise.
\end{defn}
\noindent The exponential bounds let the algorithm address a rule and
its conclusions with polynomial-length binary indices. The space
condition lets it compute the selected conclusion on demand, and the
denominator condition bounds the sizes of the rational endpoints in
recursive calls.

\begin{theorem}\label{theorem:generallypolynomialspace}
	The non-expansive logic of `generally', $\mathcal{L}_{\mathsf{gen}}$, is polynomial-space bounded.
\end{theorem}
\begin{proof}
  The condition on endpoints of conclusions holds because these are
  all either~$0$ or~$1$, or have the form $a$ or $1-a$ for an
  endpoint~$a$ of the premise. The computation of
  $\matches_{\mathcal{R}}(\Gamma)$ and
  $\concls_{\mathcal{R}}(\Gamma, i)$ as per the proof of
  \Cref{theorem:generallyexpbranching} is straightforward; we focus on
  computing $\concl_{\mathcal{R}}(\Gamma, i, j)$. Following the proofs
  of \Cref{theorem:generallyexpbounded} and
  \Cref{theorem:generallyexpbranching}, we can compute the $j$-th
  conclusion of the $i$-th modal co-rule in the following way: Iterate
  over the possible configurations of successor structures and check
  for each if it can solve the inequalities that correspond to the
  bounds of the original sequent $\Gamma$. If fewer than~$i$
  configurations pass this test, both $\concls_{\mathcal{R}}$ and
  $\concl_{\mathcal{R}}$ return~$\bot$. Otherwise take the $i$-th
  feasible configuration.  If it has fewer than~$j$ successors,
  $\concl_{\mathcal{R}}$ returns~$\bot$.  Otherwise, take the vector
  representation of the $j$-th successor state and construct the
  tableau sequent for it by using the inequalities it has to satisfy
  or not satisfy. Computing this for a configuration can be done in
  polynomial time as a linear programming problem.
\end{proof}

\noindent Since $\NPSPACE = \PSPACE$ by Savitch's theorem~\cite{Savitch70}, we then immediately obtain:

\begin{theorem}\label{thm:main}
  Let~$\mathcal{L}$ be a polynomial-space bounded logic. Then satisfiability of tableau sequents is decidable in~\PSPACE.
\end{theorem}
\begin{proof}[Proof (sketch)]
  If $F1=\emptyset$, all sequents are unsatisfiable. Otherwise, explore the recursive calls of~\Cref{alg:sat} depth-first, computing conclusions one at a time and reusing space between them. Modal depth bounds the recursion stack, while the denominator condition keeps interval endpoints polynomial in bit length. Together with polynomial-space access to co-rules, this yields a nondeterministic polynomial-space procedure, and hence the stated bound by Savitch's theorem~\cite{Savitch70}.
\end{proof}

\noindent \Cref{theorem:generallypolynomialspace} and \Cref{thm:main} then yield:
\begin{corollary}\label{theorem:generallypspace}
	Satisfiability of tableau sequents in the non-expansive logic of `generally',
	$\mathcal{L}_{\mathsf{gen}}$, is in $\PSPACE$.
\end{corollary}

\begin{example}
	Generalizing from \Cref{example:generally-corule2}, we can prove that for every rational $c\in[0,1]$, both $\Gamma_1 := \{\lnot \GENERALLY p \tin [c,c], \GENERALLY \lnot p \tin [0,c)\}$ and $\Gamma_2 := \{\lnot \GENERALLY p \tin [c,c], \GENERALLY \lnot p \tin (c,1]\}$ are unsatisfiable, i.e.~$\lnot \GENERALLY p \tin [c,c]$ implies $\GENERALLY \lnot p \tin [c,c]$. By \Cref{lemma:generallybounds}, both formulas exceed the same rational thresholds, so their values agree at every state. This obviously extends to arbitrary formulas in place of~$p$, giving us the identity $\lnot \GENERALLY \lnot \phi = \GENERALLY \phi$, i.e.~$\GENERALLY$ is self-dual (a property it shares with the operator `probably').
\end{example}

\begin{example}
  We illustrate differences in logical consequence between the logic of \emph{generally} and the logic with the expected-value operator~$\PROBABLY$ (\emph{probably}). For $f\in [0,1]^X$ and $\mu\in\Dist X$, the lifting $\llbracket\PROBABLY\rrbracket_X(f)(\mu)=\sum_{x\in X}\mu(x)f(x)$ takes the expected truth value of~$f$ under~$\mu$. For a formula~$\phi$, a rational threshold $c\in[0,1]$, and a rational constant $d\in[0,1]$, consider $\Gamma=\{\GENERALLY\phi\tin[c,c],\,\GENERALLY(\phi\land d)\tin[0,\min(c,d))\}$ and $\Gamma'=\{\PROBABLY\phi\tin[c,c],\,\PROBABLY(\phi\land d)\tin[0,\min(c,d))\}$.
  The sequent~$\Gamma$ is unsatisfiable because the values of $\GENERALLY(\phi\land d)$ and $(\GENERALLY\phi)\land d$ coincide.
On the other hand, the sequent~$\Gamma'$ can be satisfiable: Put $\phi=p$ for an atom~$p$, and let $c=0.9$ and $d=0.8$. At a state~$x$, let the successor distribution assign probabilities~$0.9$ and~$0.1$ to states~$y_1$ and~$y_2$, where~$p$ has values~$1$ and~$0$, respectively. Then $\llbracket\PROBABLY p\rrbracket(x)=0.9$ and $\llbracket(\PROBABLY p)\land 0.8\rrbracket(x)=0.8$. The argument $p\land0.8$ has values~$0.8$ and~$0$ at the two successors, so $\llbracket\PROBABLY(p\land0.8)\rrbracket(x)=0.9\cdot0.8+0.1\cdot0=0.72$.
\end{example}

\subparagraph*{Atoms} Adding propositional atoms to a
given logic as described in \Cref{nonexpfuzlog} preserves exponential
branching and polynomial-space boundedness:

\begin{lemma}\label{lemma:atoms}
	Let $\mathsf{At}$ be a set of atoms and $\mathcal{L}$ an exponentially branching / polynomial-space bounded logic. Then the logic $\mathcal{L} + \mathsf{At}$ is also exponentially branching / polynomial-space bounded.
\end{lemma}

\section{Case Studies}\label{sec:case-studies}
In addition to the non-expansive logic of \emph{generally}, we provide
three examples to which the criteria of~\Cref{polspacebound} apply,
yielding $\PSPACE$ upper bounds for the respective satisfiability
problems. In each case, we exploit a finite model property established
as per \Cref{rem:fmp}, i.e.~via a one-step finite model
property, to restrict to the \emph{finitary
  subfunctor}~$F_{\mathrm{fin}}$ of the functor~$F$ capturing the
respective system type. Here,~$F_{\mathrm{fin}}$ is given
by~$F_{\mathrm{fin}}X$ consisting of the elements of~$FX$ with finite
support; formally,
$F_{\mathrm{fin}}X=\{Fi(t)\mid Y\subseteq X\text{ finite}, i\colon
Y\hookrightarrow X,t\in FY\}$.

\subparagraph*{Non-expansive fuzzy $\ALC$}\label{nonexpansivefuzzyalc}
Non-expansive fuzzy $\ALC$~\cite{ijcai2025p502} is a fuzzy description
logic in which modalities account for both formula truth values and
the successorship degrees of successor states. For simplicity, we
restrict attention to a single role (i.e.~a fuzzy relation) and
let~$\Lambda$ consist of a single modal operator~$\Diamond$. Moreover,
we continue to keep the treatment of propositional atoms implicit as
per \Cref{polspacebound}. We interpret the logic over the (covariant)
fuzzy powerset functor~$\fPow$. For a set~$X$,
$\fPow X$ consists of the fuzzy subsets of~$X$,
i.e.~maps $A\colon X\to[0,1]$. On maps $f\colon X\to Y$, it is given by
$\fPow f(A)(y)=\sup_{x\in X\mid f(x)=y}A(x)$
for~$y\in Y$.  The coalgebras of~$\fPow$ are fuzzy
relational structures in which every pair of states~$x,y$ is assigned
a transition degree from~$x$ to~$y$. We interpret~$\Diamond$ by the
predicate lifting~$\llbracket \Diamond \rrbracket$ defined by
$\llbracket \Diamond \rrbracket_X (f)(g) = \sup_{x \in X} \min (f(x),
g(x))$ for $f\in [0,1]^X$ and~$g\in \fPow X$, in analogy
to the usual two-valued interpretation of the diamond modality.
For compact sets of admissible successor valuations, the one-step finite model property is straightforward.
As shown in the appendix, satisfiability is therefore preserved by restriction to the finite fuzzy
powerset functor~$\fPow^{\mathrm{fin}}$, where
$A(x)>0$ for only finitely many~$x$.


Over~$\fPow^{\mathrm{fin}}$, there is a single modal co-rule for each satisfiable saturated premise: It has one conclusion per nontrivial lower bound, each responsible for satisfying that bound and all relevant upper bounds. The scheme essentially captures the matches of the modal rule used in earlier work~\cite{ijcai2025p502}. By polynomial-space boundedness and~\Cref{thm:main}, we obtain that satisfiability is in~\PSPACE, recovering the known result~\cite{ijcai2025p502,Hermes23}.

\subparagraph*{Quantitative Fuzzy $\ALC$}\label{quantfuzzy} Despite
its name, the logic known as \emph{quantitative fuzzy
  $\ALC$}~\cite{SchroderPattinson11} is actually rather similar to the
logic of `generally'. Here we take~$F=\Dist$, so the coalgebras are
probabilistic transition systems. The logic features modalities~$\boldsymbol{M}_p$ for
$p \in [0,1]\cap\Rat$, read `with probability more than~$p$', in place
of the modality $\GENERALLY$.  The modalities $\boldsymbol{M}_p$ are
interpreted as the predicate liftings
\begin{equation*}
  \llbracket \boldsymbol{M}_p\rrbracket_X(f)(\mu)=\sup\{\alpha\in[0,1]\mid \mu(\{y\in X\mid f(y)\ge\alpha\})>p\}.
\end{equation*}
Thus, at a state~$x$, $\boldsymbol{M}_p\phi$ takes the supremum of the truth degrees~$\alpha$ such that a random successor of~$x$ satisfies~$\phi$ to degree at least~$\alpha$ with probability more than~$p$. Here, we take the supremum of the empty set to be~$0$.

It is known that the one-step finite model property holds for
strict bounds~\cite[Example~4.6]{SchroderPattinson11}.  In the full
version, we show the one-step finite model property for arbitrary
bounds. Thus, restricting to finitely branching probabilistic transition
systems, i.e.~to the subfunctor~$\Dist_{\mathrm{fin}}$ of distributions
with finite support, preserves satisfiability.

Over~$\Dist_{\mathrm{fin}}$, the bounds on the modalities reduce to probability inequalities, and the Boolean vector construction for \emph{generally} applies, ultimately giving us polynomial-space boundedness, so the original satisfiability problem is in~\PSPACE by~\Cref{thm:main}. For $p=a/b$ in lowest terms, we count the syntactic size of $\boldsymbol{M}_p$ as $\operatorname{bits}(a)+\operatorname{bits}(b)+1$.

\subparagraph*{Metric Modal Logic}\label{metrictracelogic} Finally, we
recall modal logics of metric transition systems with
two-valued~\cite{AlfaroEA09} and quantitative~\cite{ForsterEA25}
transition relations. We will focus on the (harder) quantitative
variant, which we just term \emph{metric modal logic}; the variant
with two-valued transitions needs only slight adjustments, detailed in
the appendix.

Let~$(L,d_L)$ be a \emph{proper} metric space of labels, i.e.~all
closed bounded balls are compact.  Following Forster
et~al.~\cite{ForsterEA25}, we take
$FX=\fPow(L\times X)$, so an $F$-coalgebra assigns to
each state pair $(x,y)$ and each label~$\ell\in L$ a degree for an
$\ell$-transition from~$x$ to~$y$. In the two-valued
variant~\cite{AlfaroEA09}, $FX=\Pow(L\times X)$ and we
identify a set of labelled successors with its characteristic
function. We use modalities
$\Lambda=\{\Diamond_\ell^c\mid\ell\in L,\ c\in[0,1]\}$. The predicate
liftings are given by
\begin{equation*}
  \llbracket \Diamond_\ell^c\rrbracket_X(f)(A)=\sup_{m\in L,\ x\in X}\min(A(m,x),f(x),c\ominus d_L(\ell,m)).
\end{equation*}
We refer to the resulting logic as \emph{metric modal logic}. For purposes of satisfiability checking, we assume that parameters~$c$ in modalities~$\Diamond_\ell^c$ are rational and that parameters~$\ell\in L$ are finitely represented (without however requiring that all elements of~$L$ are finitely representable). Intuitively, $\Diamond_\ell^c\phi$ measures the degree to which a successor with a label close to~$\ell$ satisfies~$\phi$, with~$c$ capping the contribution of label closeness.

For a singleton label space $L=\{\bullet\}$, the modality $\Diamond_\bullet^1$ coincides with the diamond of non-expansive fuzzy~$\ALC$. More generally, for rational $c\in[0,1]$, we have $\Diamond_\bullet^c\phi=c\land\Diamond\phi$.
We show in the appendix that properness ensures the one-step finite
model property for both variants. Hence, restricting to
$\fPow^{\mathrm{fin}}(L\times X)$ and $\Pow_{\mathrm{fin}}(L\times X)$,
respectively, preserves satisfiability.
We assume that distances between finitely represented elements of~$L$ are computable in polynomial space. We count the syntactic size of~$\Diamond_\ell^c$ as one plus the bit lengths of the encodings of~$\ell$ and~$c$. We require a polynomial-space test for nonemptiness of finite intersections of open or closed balls and their complements, with each centre given by a label parameter and each radius by a difference $c-a$ of a modal parameter~$c$ and a rational endpoint~$a$. Over these functors, we have a co-rule construction that enumerates, for each lower-bound witness, whether each upper bound is met by the successor value, the edge degree, or the label-distance cap, giving us polynomial-space boundedness. By~\Cref{thm:main}, satisfiability is in~\PSPACE for both variants.

\section{Conclusions and Future Work}

We have introduced the generic framework of \emph{non-expansive quantitative coalgebraic modal logic}, in which modalities in a highly general sense are combined with the non-expansive propositional base used in characteristic logics for behavioural distances~\cite{BreugelWorrell05,WildEA18,KonigMikaMichalski18,WildSchroder22} as well as in the recently introduced fuzzy description logic non-expansive fuzzy $\ALC$~\cite{ijcai2025p502}. We provide a criterion that guarantees decidability of threshold satisfiability in \PSPACE. By instantiation of this result, we have obtained new \PSPACE upper bounds for reasoning in a range of concrete instances, including modal logics of metric transition systems under compactness and computability assumptions on the label space as well as two non-expansive probabilistic modal logics, specifically \emph{quantitative fuzzy $\ALC$} and the logic of \emph{generally}~\cite{SchroderPattinson11}, in the latter cases complementing known $\NEXP$ upper bounds for strict satisfiability in the respective \L{}ukasiewicz variants. Our generic criterion works by reduction of the full logic to the so-called one-step logic. Notably, our criterion involves an \emph{exponential branching} condition requiring essentially that truth values of modal arguments on successor states can vary independently within specified bounds from exponential-sized sets of such bounds.

An important remaining open issue on the side of concrete instance logics concerns the logic of \emph{probably}~\cite{SchroderPattinson11}, in which `probably' is interpreted as taking expected truth values, following Zadeh~\cite{Zadeh68} and H{\'{a}}jek~\cite{Hajek07}. The best known upper bound for strict satisfiability in the logic of \emph{probably} over the full \L{}ukasiewicz base is \textsc{ExpSpace}~\cite{SchroderPattinson11} (H{\'{a}}jek~\cite{Hajek07} proves a \PSPACE upper bound for the fragment without nested modalities). We show in the appendix how shared successor probabilities prevent a finite co-rule scheme for \emph{probably}. The problem of giving a better upper complexity bound for the modal logic of \emph{probably} over the non-expansive propositional base, i.e.~for van Breugel and Worrell's characteristic logic for behavioural distance on probabilistic transition systems under the Kantorovich lifting~\cite{BreugelWorrell05}, thus remains open (the logic of \emph{generally} is, instead, characteristic for $\epsilon$-distance, or \emph{L\'evy-Prokhorov distance}, on probabilistic transition systems~\cite{DesharnaisEA08,DesharnaisSokolova26,WildEA26}). A further point for future research is to obtain a generic algorithm for reasoning with global assumptions (TBoxes) in non-expansive quantitative coalgebraic modal logic, ideally realizing an upper bound $\EXP$ as in the base case of non-expansive fuzzy $\ALC$~\cite{ijcai2025p502}.

\bibliographystyle{plainurl}

\bibliography{coalg-logic}

\begin{thebibliography}{10}

\bibitem{AdamekEA91}
Jir{\'{\i}} Ad{\'{a}}mek, Horst Herrlich, and George~E. Strecker.
\newblock {\em Abstract and Concrete Categories -- The Joy of Cats}.
\newblock Wiley, 1991.
\newblock Reprint in {\emph{Theory Appl.\ Cat.}} 17 (2006).
\newblock URL:
  \url{http://www.tac.mta.ca/tac/reprints/articles/17/tr17abs.html}.

\bibitem{AlurEA02}
Rajeev Alur, Thomas~A. Henzinger, and Orna Kupferman.
\newblock Alternating-time temporal logic.
\newblock {\em Journal of the {ACM}}, 49(5):672--713, 2002.
\newblock \href {https://doi.org/10.1145/585265.585270}
  {\path{doi:10.1145/585265.585270}}.

\bibitem{AolariteiEA25}
Liviu Aolaritei, Julie Zhu, Zheyu~Oliver Wang, Michael~I. Jordan, and Youssef
  Marzouk.
\newblock Conformal prediction under {L}{\'{e}}vy-{P}rokhorov distribution
  shifts: Robustness to local and global perturbations.
\newblock In {\em Neural Information Processing Systems 2025, NeurIPS 2025},
  2025.
\newblock URL:
  \url{http://papers.nips.cc/paper\_files/paper/2025/hash/5c3760faa9de12a0ff5afb839581a09f-Abstract-Conference.html}.

\bibitem{Baader2015}
Franz Baader, Stefan Borgwardt, and Rafael Pe{\~{n}}aloza.
\newblock On the decidability status of fuzzy $\mathcal{ALC}$ with general
  concept inclusions.
\newblock {\em Journal of Philosophical Logic}, 44(2):117--146, Apr 2015.
\newblock \href {https://doi.org/10.1007/s10992-014-9329-3}
  {\path{doi:10.1007/s10992-014-9329-3}}.

\bibitem{BaldanEA18}
Paolo Baldan, Filippo Bonchi, Henning Kerstan, and Barbara K{\"{o}}nig.
\newblock Coalgebraic behavioral metrics.
\newblock {\em Logical Methods in Computer Science}, 14(3), 2018.
\newblock \href {https://doi.org/10.23638/LMCS-14(3:20)2018}
  {\path{doi:10.23638/LMCS-14(3:20)2018}}.

\bibitem{BennounaEA23}
M.~Amine Bennouna, Ryan Lucas, and Bart P. G.~Van Parys.
\newblock Certified robust neural networks: Generalization and corruption
  resistance.
\newblock In Andreas Krause, Emma Brunskill, Kyunghyun Cho, Barbara Engelhardt,
  Sivan Sabato, and Jonathan Scarlett, editors, {\em International Conference
  on Machine Learning, {ICML} 2023}, volume 202 of {\em Proc.\ Machine Learning
  Res.}, pages 2092--2112. {PMLR}, 2023.
\newblock URL: \url{https://proceedings.mlr.press/v202/bennouna23a.html}.

\bibitem{BonattiTettamanzi03}
Piero~A. Bonatti and Andrea Tettamanzi.
\newblock Some complexity results on fuzzy description logics.
\newblock In Vito~Di Ges{\`{u}}, Francesco Masulli, and Alfredo Petrosino,
  editors, {\em Fuzzy Logic and Applications, {WILF} 2003}, volume 2955 of {\em
  LNCS}, pages 19--24. Springer, 2003.
\newblock \href {https://doi.org/10.1007/10983652\_3}
  {\path{doi:10.1007/10983652\_3}}.

\bibitem{BorgwardtEA16}
Stefan Borgwardt and Rafael Pe{\~{n}}aloza.
\newblock Reasoning in expressive {G}{\"{o}}del description logics.
\newblock In Maurizio Lenzerini and Rafael Pe{\~{n}}aloza, editors, {\em
  Description Logics, DL 2016}, volume 1577 of {\em {CEUR} Workshop
  Proceedings}. CEUR-WS.org, 2016.
\newblock URL: \url{https://ceur-ws.org/Vol-1577/paper\_2.pdf}.

\bibitem{BouEA11}
F{\'{e}}lix Bou, Marco Cerami, and Francesc Esteva.
\newblock Finite-valued {{\L{}}}ukasiewicz modal logic is {PSPACE}-complete.
\newblock In Toby Walsh, editor, {\em International Joint Conference on
  Artificial Intelligence, {IJCAI} 2011}, pages 774--779. {IJCAI/AAAI}, 2011.
\newblock \href {https://doi.org/10.5591/978-1-57735-516-8/IJCAI11-136}
  {\path{doi:10.5591/978-1-57735-516-8/IJCAI11-136}}.

\bibitem{CaicedoEA17}
Xavier Caicedo, George Metcalfe, Ricardo~Oscar Rodr{\'{\i}}guez, and Jonas
  Rogger.
\newblock Decidability of order-based modal logics.
\newblock {\em Journal of Computer and System Sciences}, 88:53--74, 2017.
\newblock \href {https://doi.org/10.1016/j.jcss.2017.03.012}
  {\path{doi:10.1016/j.jcss.2017.03.012}}.

\bibitem{CeramiEsteva22}
Marco Cerami and Francesc Esteva.
\newblock On decidability of concept satisfiability in description logic with
  product semantics.
\newblock {\em Fuzzy Sets and Systems}, 445:1--21, 2022.
\newblock \href {https://doi.org/10.1016/j.fss.2021.11.013}
  {\path{doi:10.1016/j.fss.2021.11.013}}.

\bibitem{Chellas80}
Brian~F. Chellas.
\newblock {\em Modal Logic - An Introduction}.
\newblock Cambridge University Press, 1980.
\newblock \href {https://doi.org/10.1017/CBO9780511621192}
  {\path{doi:10.1017/CBO9780511621192}}.

\bibitem{CirsteaEA11}
Corina C{\^{\i}}rstea, Clemens Kupke, and Dirk Pattinson.
\newblock {EXPTIME} tableaux for the coalgebraic {$\mu$}-calculus.
\newblock {\em Logical Methods in Computer Science}, 7(3), 2011.
\newblock \href {https://doi.org/10.2168/LMCS-7(3:3)2011}
  {\path{doi:10.2168/LMCS-7(3:3)2011}}.

\bibitem{AlfaroEA09}
Luca de~Alfaro, Marco Faella, and Mari{\"{e}}lle Stoelinga.
\newblock Linear and branching system metrics.
\newblock {\em {IEEE} Transactions on Software Engineering}, 35(2):258--273,
  2009.
\newblock \href {https://doi.org/10.1109/TSE.2008.106}
  {\path{doi:10.1109/TSE.2008.106}}.

\bibitem{DesharnaisEA08}
Jos{\'{e}}e Desharnais, Fran{\c{c}}ois Laviolette, and Mathieu Tracol.
\newblock Approximate analysis of probabilistic processes: Logic, simulation
  and games.
\newblock In {\em Quantitative Evaluaiton of Systems, {QEST} 2008}, pages
  264--273. {IEEE} Computer Society, 2008.
\newblock \href {https://doi.org/10.1109/QEST.2008.42}
  {\path{doi:10.1109/QEST.2008.42}}.

\bibitem{DesharnaisSokolova26}
Jos{\'{e}}e Desharnais and Ana Sokolova.
\newblock {\(\epsilon\)}-distance via l{\'{e}}vy-prokhorov lifting.
\newblock In Stefano Guerrini and Barbara König, editors, {\em Computer
  Science Logic, CSL 2026}, LIPIcs. Schloss Dagstuhl -- Leibniz-Zentrum f{\"u}r
  Informatik, 2026.
\newblock To appear; available on arXiv under
  {\url{https://arxiv.org/abs/2507.10732}}.

\bibitem{DuEA16}
Wenjie Du, Yuxin Deng, and Daniel Gebler.
\newblock Behavioural pseudometrics for nondeterministic probabilistic systems.
\newblock In Martin Fr{\"{a}}nzle, Deepak Kapur, and Naijun Zhan, editors, {\em
  Dependable Software Engineering: Theories, Tools, and Applications, {SETTA}
  2016}, volume 9984 of {\em LNCS}, pages 67--84. Springer, 2016.
\newblock \href {https://doi.org/10.1007/978-3-319-47677-3}
  {\path{doi:10.1007/978-3-319-47677-3}}.

\bibitem{EleftheriouEA12}
Pantelis Eleftheriou, Costas Koutras, and Christos Nomikos.
\newblock Notions of bisimulation for {H}eyting-valued modal languages.
\newblock {\em Journal of Logic and Computation}, 22(2):213--235, 2012.
\newblock \href {https://doi.org/10.1093/logcom/exq005}
  {\path{doi:10.1093/logcom/exq005}}.

\bibitem{Fan15}
Tuan{-}Fang Fan.
\newblock Fuzzy bisimulation for {G}{\"{o}}del modal logic.
\newblock {\em {IEEE} Transactions on Fuzzy Systems}, 23(6):2387--2396, 2015.
\newblock \href {https://doi.org/10.1109/TFUZZ.2015.2426724}
  {\path{doi:10.1109/TFUZZ.2015.2426724}}.

\bibitem{ForsterEA25}
Jonas Forster, Lutz Schr{\"{o}}der, Paul Wild, Harsh Beohar, Sebastian Gurke,
  Barbara K{\"{o}}nig, and Karla Messing.
\newblock Quantitative graded semantics and spectra of behavioural metrics.
\newblock In J{\"{o}}rg Endrullis and Sylvain Schmitz, editors, {\em Computer
  Science Logic, {CSL} 2025}, volume 326 of {\em LIPIcs}, pages 33:1--33:21.
  Schloss Dagstuhl -- Leibniz-Zentrum f{\"{u}}r Informatik, 2025.
\newblock \href {https://doi.org/10.4230/LIPICS.CSL.2025.33}
  {\path{doi:10.4230/LIPICS.CSL.2025.33}}.

\bibitem{ForsterEA26}
Jonas Forster, Lutz Schr{\"{o}}der, Paul Wild, Barbara K{\"{o}}nig, and Pedro
  Nora.
\newblock Threshold-based behavioural distances.
\newblock In Ana Sokolova and Patrick Totzke, editors, {\em 37th International
  Conference on Concurrency Theory (CONCUR 2026)}, volume 391 of {\em LIPIcs},
  pages 35:1--35:19. Schloss Dagstuhl -- Leibniz-Zentrum f{\"u}r Informatik,
  2026.
\newblock \href {https://doi.org/10.4230/LIPIcs.CONCUR.2026.35}
  {\path{doi:10.4230/LIPIcs.CONCUR.2026.35}}.

\bibitem{ijcai2025p502}
Stefan Gebhart, Lutz Schröder, and Paul Wild.
\newblock Non-expansive fuzzy {$\mathcal{ALC}$}.
\newblock In James Kwok, editor, {\em Proceedings of the Thirty-Fourth
  International Joint Conference on Artificial Intelligence, {IJCAI-25}}, pages
  4509--4517. International Joint Conferences on Artificial Intelligence
  Organization, 8 2025.
\newblock Main Track.
\newblock \href {https://doi.org/10.24963/ijcai.2025/502}
  {\path{doi:10.24963/ijcai.2025/502}}.

\bibitem{Hajek07}
Petr H{\'{a}}jek.
\newblock Complexity of fuzzy probability logics {II}.
\newblock {\em Fuzzy Sets and Systems}, 158(23):2605--2611, 2007.
\newblock \href {https://doi.org/10.1016/J.FSS.2007.07.001}
  {\path{doi:10.1016/J.FSS.2007.07.001}}.

\bibitem{HausmannSchroder24}
Daniel Hausmann and Lutz Schr{\"{o}}der.
\newblock Coalgebraic satisfiability checking for arithmetic {\(\mu\)}-calculi.
\newblock {\em Logical Methods in Computer Science}, 20(3), 2024.
\newblock \href {https://doi.org/10.46298/LMCS-20(3:9)2024}
  {\path{doi:10.46298/LMCS-20(3:9)2024}}.

\bibitem{Hermes23}
Philipp Hermes.
\newblock Tableau calculi for non-expansive fuzzy modal logics, 2023.
\newblock Bachelor's thesis, Friedrich-Alexander-Universität
  Erlangen-Nürnberg.

\bibitem{HuthKwiatkowska97}
Michael Huth and Marta Kwiatkowska.
\newblock Quantitative analysis and model checking.
\newblock In {\em Logic in Computer Science, LICS 1997}, pages 111--122. IEEE,
  1997.
\newblock \href {https://doi.org/10.1109/LICS.1997.614940}
  {\path{doi:10.1109/LICS.1997.614940}}.

\bibitem{Hajek95}
Petr Hájek.
\newblock Fuzzy logic and arithmetical hierarchy.
\newblock {\em Fuzzy Sets and Systems}, 73(3):359--363, 1995.
\newblock \href {https://doi.org/https://doi.org/10.1016/0165-0114(94)00299-M}
  {\path{doi:https://doi.org/10.1016/0165-0114(94)00299-M}}.

\bibitem{KellerHeymans09}
Uwe Keller and Stijn Heymans.
\newblock Fuzzy description logic reasoning using a fixpoint algorithm.
\newblock In {\em Logical Foundations of Computer Science, LFCS 2009}, volume
  5407 of {\em LNCS}, pages 265--279. Springer, 2009.
\newblock \href {https://doi.org/10.1007/978-3-540-92687-0}
  {\path{doi:10.1007/978-3-540-92687-0}}.

\bibitem{KonigMikaMichalski18}
Barbara K{\"{o}}nig and Christina Mika{-}Michalski.
\newblock ({M}etric) bisimulation games and real-valued modal logics for
  coalgebras.
\newblock In Sven Schewe and Lijun Zhang, editors, {\em Concurrency Theory,
  {CONCUR} 2018}, volume 118 of {\em LIPIcs}, pages 37:1--37:17. Schloss
  Dagstuhl -- Leibniz-Zentrum f{\"{u}}r Informatik, 2018.
\newblock \href {https://doi.org/10.4230/LIPICS.CONCUR.2018.37}
  {\path{doi:10.4230/LIPICS.CONCUR.2018.37}}.

\bibitem{KulackaEA13}
Agnieszka Kulacka, Dirk Pattinson, and Lutz Schr{\"{o}}der.
\newblock Syntactic labelled tableaux for {{\L{}}}ukasiewicz fuzzy {ALC}.
\newblock In Francesca Rossi, editor, {\em International Joint Conference on
  Artificial Intelligence, {IJCAI} 2013}, pages 962--968. {IJCAI/AAAI}, 2013.
\newblock URL: \url{https://www.ijcai.org/Proceedings/13/Papers/147.pdf}.

\bibitem{Kupke2022}
Clemens Kupke, Dirk Pattinson, and Lutz Schr\"{o}der.
\newblock Coalgebraic reasoning with global assumptions in arithmetic modal
  logics.
\newblock {\em {ACM} Transactions on Computational Logic}, 23(2), January 2022.
\newblock \href {https://doi.org/10.1145/3501300} {\path{doi:10.1145/3501300}}.

\bibitem{LarsenSkou91}
Kim~Guldstrand Larsen and Arne Skou.
\newblock Bisimulation through probabilistic testing.
\newblock {\em Information and Computation}, 94(1):1--28, 1991.
\newblock \href {https://doi.org/10.1016/0890-5401(91)90030-6}
  {\path{doi:10.1016/0890-5401(91)90030-6}}.

\bibitem{LukasiewiczStraccia08}
Thomas Lukasiewicz and Umberto Straccia.
\newblock Managing uncertainty and vagueness in description logics for the
  semantic web.
\newblock {\em Journal of Web Semantics}, 6(4):291--308, 2008.
\newblock \href {https://doi.org/10.1016/J.WEBSEM.2008.04.001}
  {\path{doi:10.1016/J.WEBSEM.2008.04.001}}.

\bibitem{MorganMcIver97}
Carroll Morgan and Annabelle McIver.
\newblock A probabilistic temporal calculus based on expectations.
\newblock In Lindsay Groves and Steve Reeves, editors, {\em Formal Methods
  Pacific, FMP 1997}. Springer, 1997.

\bibitem{Pattinson04}
Dirk Pattinson.
\newblock Expressive logics for coalgebras via terminal sequence induction.
\newblock {\em Notre Dame Journal of Formal Logic}, 45(1):19--33, 2004.
\newblock \href {https://doi.org/10.1305/NDJFL/1094155277}
  {\path{doi:10.1305/NDJFL/1094155277}}.

\bibitem{Pauly02}
Marc Pauly.
\newblock A modal logic for coalitional power in games.
\newblock {\em Journal of Logic and Computation}, 12(1):149--166, 2002.
\newblock \href {https://doi.org/10.1093/LOGCOM/12.1.149}
  {\path{doi:10.1093/LOGCOM/12.1.149}}.

\bibitem{Pavelka79}
J.\ Pavelka.
\newblock On fuzzy logic {I}, {II}, {III}.
\newblock {\em Zeitschrift f{\"u}r mathematische Logik und Grundlagen der
  Mathematik}, 45:45--52, 119--134, 447--464, 1979.

\bibitem{Rutten00}
Jan J. M.~M. Rutten.
\newblock Universal coalgebra: a theory of systems.
\newblock {\em Theoretical Computer Science}, 249(1):3--80, 2000.
\newblock \href {https://doi.org/10.1016/S0304-3975(00)00056-6}
  {\path{doi:10.1016/S0304-3975(00)00056-6}}.

\bibitem{Savitch70}
Walter~J. Savitch.
\newblock Relationships between nondeterministic and deterministic tape
  complexities.
\newblock {\em Journal of Computer and System Sciences}, 4(2):177--192, 1970.
\newblock \href {https://doi.org/10.1016/S0022-0000(70)80006-X}
  {\path{doi:10.1016/S0022-0000(70)80006-X}}.

\bibitem{Schroder07}
Lutz Schr{\"{o}}der.
\newblock A finite model construction for coalgebraic modal logic.
\newblock {\em The Journal of Logic and Algebraic Programming},
  73(1-2):97--110, 2007.
\newblock \href {https://doi.org/10.1016/J.JLAP.2006.11.004}
  {\path{doi:10.1016/J.JLAP.2006.11.004}}.

\bibitem{Schroder08}
Lutz Schr{\"{o}}der.
\newblock Expressivity of coalgebraic modal logic: The limits and beyond.
\newblock {\em Theoretical Computer Science}, 390(2-3):230--247, 2008.
\newblock \href {https://doi.org/10.1016/J.TCS.2007.09.023}
  {\path{doi:10.1016/J.TCS.2007.09.023}}.

\bibitem{SchroderPattinson08}
Lutz Schr{\"{o}}der and Dirk Pattinson.
\newblock Shallow models for non-iterative modal logics.
\newblock In Andreas Dengel, Karsten Berns, Thomas~M. Breuel, Frank Bomarius,
  and Thomas Roth{-}Berghofer, editors, {\em Advances in Artificial
  Intelligence, {KI} 2008}, volume 5243 of {\em LNCS}, pages 324--331.
  Springer, 2008.
\newblock \href {https://doi.org/10.1007/978-3-540-85845-4\_40}
  {\path{doi:10.1007/978-3-540-85845-4\_40}}.

\bibitem{SchroderPattinson09}
Lutz Schr{\"{o}}der and Dirk Pattinson.
\newblock {PSPACE} bounds for rank-1 modal logics.
\newblock {\em {ACM} Transactions on Computational Logic}, 10(2):13:1--13:33,
  2009.
\newblock \href {https://doi.org/10.1145/1462179.1462185}
  {\path{doi:10.1145/1462179.1462185}}.

\bibitem{SchroderPattinson11}
Lutz Schr{\"o}der and Dirk Pattinson.
\newblock Description logics and fuzzy probability.
\newblock In Toby Walsh, editor, {\em International Joint Conference on
  Artificial Intelligence, IJCAI 2011}, pages 1075--1081. AAAI Press, 2011.
\newblock \href {https://doi.org/10.5591/978-1-57735-516-8/IJCAI11-184}
  {\path{doi:10.5591/978-1-57735-516-8/IJCAI11-184}}.

\bibitem{StoilosEA07}
Giorgos Stoilos, Giorgos Stamou, Jeff Pan, Vassilis Tzouvaras, and Ian
  Horrocks.
\newblock Reasoning with very expressive fuzzy description logics.
\newblock {\em Journal of Artificial Intelligence Research}, 30:273--320, 2007.
\newblock \href {https://doi.org/10.1613/jair.2279}
  {\path{doi:10.1613/jair.2279}}.

\bibitem{StoilosEA14}
Giorgos Stoilos and Giorgos~B. Stamou.
\newblock Reasoning with fuzzy extensions of {OWL} and {OWL} 2.
\newblock {\em Knowledge and Information Systems}, 40(1):205--242, 2014.
\newblock \href {https://doi.org/10.1007/S10115-013-0641-Y}
  {\path{doi:10.1007/S10115-013-0641-Y}}.

\bibitem{Straccia01}
Umberto Straccia.
\newblock Reasoning within fuzzy description logics.
\newblock {\em Journal of Artificial Intelligence Research}, 14:137--166, 2001.
\newblock \href {https://doi.org/10.1613/JAIR.813}
  {\path{doi:10.1613/JAIR.813}}.

\bibitem{Straccia05}
Umberto Straccia.
\newblock Description logics with fuzzy concrete domains.
\newblock In {\em Uncertainty in Artificial Intelligence, UAI 2005}, pages
  559--567. AUAI Press, 2005.

\bibitem{StracciaBobillo07}
Umberto Straccia and Fernando Bobillo.
\newblock Mixed integer programming, general concept inclusions and fuzzy
  description logics.
\newblock In {\em New Dimensions in Fuzzy Logic and Related Technologies,
  EUSFLAT 2007}, pages 213--220. Universitas Ostraviensis, 2007.

\bibitem{1975TheoryOF}
Michio Sugeno.
\newblock {\em Theory of fuzzy integrals and its applications}.
\newblock PhD thesis, Tokyo Institute of Technology, 1974.

\bibitem{BreugelWorrell05}
F.~{van Breugel} and J.~Worrell.
\newblock A behavioural pseudometric for probabilistic transition systems.
\newblock {\em Theoretical Computer Science}, 331(1):115--142, 2005.
\newblock \href {https://doi.org/10.1016/j.tcs.2004.09.035}
  {\path{doi:10.1016/j.tcs.2004.09.035}}.

\bibitem{Vardi96}
Moshe~Y. Vardi.
\newblock Why is modal logic so robustly decidable?
\newblock In Neil Immerman and Phokion~G. Kolaitis, editors, {\em Descriptive
  Complexity and Finite Models, Proceedings of a {DIMACS} Workshop 1996},
  volume~31 of {\em {DIMACS} Series in Discrete Mathematics and Theoretical
  Computer Science}, pages 149--183. {DIMACS/AMS}, 1996.
\newblock \href {https://doi.org/10.1090/DIMACS/031/05}
  {\path{doi:10.1090/DIMACS/031/05}}.

\bibitem{WildSchroder22}
Paul Wild and Lutz Schr{\"{o}}der.
\newblock Characteristic logics for behavioural hemimetrics via fuzzy lax
  extensions.
\newblock {\em Logical Methods in Computer Science}, 18(2), 2022.
\newblock \href {https://doi.org/10.46298/LMCS-18(2:19)2022}
  {\path{doi:10.46298/LMCS-18(2:19)2022}}.

\bibitem{WildEA26}
Paul Wild, Lutz Schr{\"{o}}der, Karla Messing, Barbara K{\"{o}}nig, and Jonas
  Forster.
\newblock Generalized {K}antorovich-{R}ubinstein duality beyond {H}ausdorff and
  {K}antorovich.
\newblock In Nathalie Bertrand and Stefan Milius, editors, {\em Foundations of
  Software Science and Computation Structures, FoSSaCS 2026}, volume 16503 of
  {\em LNCS}, pages 636--657. Springer, 2026.
\newblock \href {https://doi.org/10.1007/978-3-032-22730-0\_30}
  {\path{doi:10.1007/978-3-032-22730-0\_30}}.

\bibitem{WildEA18}
Paul Wild, Lutz Schr{\"{o}}der, Dirk Pattinson, and Barbara K{\"{o}}nig.
\newblock A van {B}enthem theorem for fuzzy modal logic.
\newblock In Anuj Dawar and Erich Gr{\"{a}}del, editors, {\em Logic in Computer
  Science, {LICS} 2018}, pages 909--918. {ACM}, 2018.
\newblock \href {https://doi.org/10.1145/3209108.3209180}
  {\path{doi:10.1145/3209108.3209180}}.

\bibitem{Zadeh68}
Lotfi~A.\ Zadeh.
\newblock Probability measures of fuzzy events.
\newblock {\em Journal of Mathematical Analysis and Applications},
  23(2):421--427, 1968.
\newblock \href {https://doi.org/https://doi.org/10.1016/0022-247X(68)90078-4}
  {\path{doi:https://doi.org/10.1016/0022-247X(68)90078-4}}.

\bibitem{ZadehAliev18}
Lotfi~A. Zadeh and Rafik~A. Aliev.
\newblock {\em Fuzzy Logic Theory and Applications - Part {I} and Part {II}}.
\newblock WorldScientific, 2018.
\newblock \href {https://doi.org/10.1142/10936} {\path{doi:10.1142/10936}}.

\end{thebibliography}

\newpage

\appendix 
\section{Appendix: Additional Details and Omitted Proofs}\label{appendix}
In this section, we collect some of the more technical details of the proofs of the main body.

\begin{convention}
  Throughout, we use placeholders
  $\triangleleft,\triangleright,\bowtie$ with ranges
  \begin{equation*}
    \triangleleft  \in  \{<, \leq\} \qquad
    \triangleright  \in  \{>, \geq\}  \qquad
    \bowtie  \in  \{<, \leq, >, \geq\}.
  \end{equation*}
\end{convention}

\begin{defn}
	We define the following operators:
	\begin{equation*}
		\bowtie^\circ :=
		\begin{cases}
			<, & \text{if } {\bowtie}={>}\\
			\leq, &\text{if }{\bowtie}={\geq}\\
			>, &\text{if }{\bowtie}={<}\\
			\geq, &\text{if }{\bowtie}={\leq}
		\end{cases}
		\qquad
		\overline{\bowtie} :=
		\begin{cases}
			\geq, &\text{if }{\bowtie}={>}\\
			>, &\text{if }{\bowtie}={\geq}\\
			\leq, &\text{if }{\bowtie}={<}\\
			<, &\text{if }{\bowtie}={\leq}
		\end{cases}
	\end{equation*}
\end{defn}

\begin{defn}
	Let $L$ be a set of \emph{labels}. For exact tableau sequents $\Gamma$, $\Gamma'$ over $L$, we say that $\Gamma'$ is a \emph{sub-sequent} of $\Gamma$ if for all $\ell \in L$ we have $\Gamma'(\ell) \subseteq \Gamma(\ell)$.
\end{defn}

\begin{example}\label{example:sequents}
	\begin{enumerate}
		\item Let $p$ be a modality emulating an atom. Then the tableau sequent $(\lnot p) \lor p \tin [0, 0.5)$ is obviously not satisfiable, which means its negation $(\lnot p) \lor p \tin [0.5, 1]$ is valid, i.e.~it is satisfied by all models and states. For clarity, we omitted the inner formula for the modality $p$.
		\item Recall the formulas $\mathsf{grave\_injury}$ and $\boldsymbol{M}_{0.9}\,\mathsf{recovery}$ from non-expansive quantitative fuzzy $\ALC$. Then the tableau sequent $\Gamma := \{{\mathsf{grave\_injury} \tin [0.6,1]}, {\boldsymbol{M}_{0.9}\,\mathsf{recovery} \tin [0.8,1]}\}$ is satisfiable, meaning that there can be grave injuries that nevertheless have a $90\%$ chance that the patient will make a good recovery. 
	\end{enumerate}
\end{example}

\begin{proof}[Proof of \Cref{localreduction}]
	If $\Gamma$ is satisfiable, there exists a state $x$ in some coalgebra $M=(X, \xi)$ that satisfies $\Gamma$. Clearly, we then also have a one-step model with $t$ given as $\xi(x)$, the successor structure of~$x$, and $\tau$ taking the respective values from $M$ for each formula and state. It is then trivial that for each $y \in X$, the required values are jointly realized by~$y$ in~$M$.
	On the other hand, let $M'=(X, \tau, t)$ be a one-step model satisfying $\Gamma^\sharp$ and suppose that for each $x \in X$ the values $\tau(v)(x)$, $v\in V$, are jointly realizable. Then we obtain coalgebras $M_x = (Y_x, \xi_x)$ realizing these values in a state which we also name $x \in Y_x$. We then obtain a full coalgebra $M=(Y, \xi)$ in the following way: Put $Y := \{y\} \mathbin{\dot{\cup}}_{x \in X} Y_x$ as the disjoint union of the coalgebras and a fresh state $y$. We can then define the coalgebra structure $\xi\colon Y \rightarrow FY$ as the universal coproduct morphism for the family $F \iota_X \circ t!\colon \{y\} \rightarrow FX \rightarrow FY$, $F \iota_{Y_x} \circ \xi_x\colon Y_x \rightarrow F Y_x \rightarrow FY$.
	Because all $M_x$ are now subcoalgebras of~$M$, we have that $\llbracket \rho \rrbracket_M (x) = \llbracket \rho \rrbracket_{M_x} (x)$ for all formulas $\rho$. This means in particular we have $\llbracket (\sigma_\Gamma v) \rrbracket_M \vert_X = \tau(v)$.
	We then use this fact and the naturality of predicate liftings:
	\begin{align*}
		\llbracket \heartsuit (\sigma_\Gamma v) \rrbracket_M (y) &= \llbracket \heartsuit \rrbracket_Y (\llbracket (\sigma_\Gamma v) \rrbracket_M)(F\iota_X(t))\\
		&= \llbracket \heartsuit \rrbracket_X (\tau(v)) t
	\end{align*}
	This then in turn directly implies that $\Gamma$ is satisfied by $y$ in this coalgebra.
\end{proof}

\begin{example}
	We construct some full tableaux for tableau sequents, where each branch of a full tableau is a propositional tableau.
	\begin{enumerate}
		\item We start by constructing a full tableau for the first example in \ref{example:sequents}.
		Note that $(\lnot p) \lor p \tin [0, 0.5) = \lnot (p \land (\lnot p)) \tin [0, 0.5)$.
		{
			\begin{gather*}
				\cfrac{\lnot (p \land (\lnot p)) \tin [0, 0.5)}{}(\lnot)\\
				\cfrac{p \land (\lnot p) \tin (0.5, 1]}{}(\land)\\
				\cfrac{p \tin (0.5, 1], \lnot p \tin (0.5, 1]}{}(\lnot) \hspace*{6mm} \cfrac{p \tin (0.5, 1], \lnot p \tin (0.5, 1]}{}(\lnot)\\
				\cfrac{p \tin (0.5, 1], p \tin [0, 0.5)}{}(\cap )\hspace*{6mm}\cfrac{p \tin (0.5, 1], p \tin [0, 0.5)}{}(\cap)\\
				\cfrac{p \tin \emptyset}{\bot}(\text{Ax})\hspace*{30mm}\cfrac{p \tin \emptyset}{\bot}(\text{Ax})
			\end{gather*}
		}
		\item We continue with an example that highlights how the constant shift rules $(\ominus)$ and $(\ominus)'$ work.
		{
			\begin{gather*}
				\cfrac{
					\begin{array}{l}
						\lnot (p \ominus 0.2) \land \lnot (0 \ominus 0.1) \tin [0.4, 0.9],\\
						p \tin [0,0.2]
					\end{array}
				}{}(\land)\\[2mm]
				\cfrac{
					\begin{array}{l}
						\lnot (p \ominus 0.2) \tin [0.4,0.9],\\
						\lnot (0 \ominus 0.1) \tin [0.4, 1],\\
						p \tin [0,0.2]
					\end{array}
				}{}(\lnot)
				\hspace*{6mm}
				\cfrac{
					\begin{array}{l}
						\lnot (p \ominus 0.2) \tin [0.4,1],\\
						\lnot (0 \ominus 0.1) \tin [0.4, 0.9],\\
						p \tin [0,0.2]
					\end{array}
				}{}(\lnot)\\[2mm]			
				\cfrac{
					\begin{array}{l}
						(p \ominus 0.2) \tin [0.1,0.6],\\
						(0 \ominus 0.1) \tin [0, 0.6],\\
						p \tin [0,0.2]
					\end{array}
				}{}(\ominus)
				\hspace*{6mm}
				\cfrac{
					\begin{array}{l}
						(p \ominus 0.2) \tin [0,0.6],\\
						(0 \ominus 0.1) \tin [0.1, 0.6],\\
						p \tin [0,0.2]
					\end{array}
				}{}(\ominus')\\[2mm]			
				\cfrac{
					\begin{array}{l}
						p \tin [0.3,0.8],\\
						(0 \ominus 0.1) \tin [0, 0.6],\\
						p \tin [0,0.2]
					\end{array}
				}{}(\ominus')
				\hspace*{6mm}
				\cfrac{
					\begin{array}{l}
						p \tin [0,0.8],\\
						(0 \ominus 0.1) \tin [0.1, 0.6],\\
						p \tin [0,0.2]
					\end{array}
				}{}(\ominus)\\[2mm]			
				\cfrac{
					\begin{array}{l}
						p \tin [0.3,0.8],\\
						0 \tin [0, 0.7],\\
						p \tin [0,0.2]
					\end{array}
				}{}(\cap)
				\hspace*{6mm}
				\cfrac{
					\begin{array}{l}
						p \tin [0,0.8],\\
						0 \tin [0.2, 0.7],\\
						p \tin [0,0.2]
					\end{array}
				}{\bot}(\text{Ax 0})\\[2mm]
				\cfrac{
					\begin{array}{l}
						p \tin \emptyset,\\
						0 \tin [0, 0.7]
					\end{array}
				}{\bot}(\text{Ax})
				\hspace*{35mm}
			\end{gather*}
		}
	\end{enumerate}
\end{example}

\begin{proof}[Proof of \Cref{tableau:correctness}]
	Let $\Gamma$ be a tableau sequent over a set $L$ of one-step formulas. 
	If there is an open propositional tableau $G$ for~$\Gamma$, then it is immediately clear that $\Gamma_G$ is an exact tableau sequent over a set of formulas of the form $\heartsuit v \in \mathsf{S}_0(L)$, since otherwise a rule would still be applicable and $G$ would not be a propositional tableau.
	In what follows, we identify $\Gamma_\bot$ as an exact tableau sequent that is never satisfiable by a one-step model, e.g. the sequent $0 \tin \emptyset$.
	Let $\Gamma$ be a sequent and let $C$ or $C_1, C_2$ be the conclusion(s) of a rule whose premise matches~$\Gamma$. We show that a one-step model $M$ satisfies $\Gamma$ iff it satisfies $C$ or at least one of $C_1$ and $C_2$.
	We investigate each rule separately:
	\begin{itemize}
		\item $(\text{Ax})$: If $\Gamma$ matches $S, \phi \tin \emptyset$, then $\Gamma (\phi) = \emptyset$ and hence~$\Gamma$ is never satisfiable.
		\item $(\text{Ax} 0)$: If $\Gamma$ matches $S, 0 \tin I$ with $0 \notin I$, then $\Gamma (0) = J \subseteq I$ and $\Gamma$ is never satisfiable because $\llbracket 0 \rrbracket = 0$ in any one-step model.
		\item $(\cap)$: In this case, satisfiability of~$\Gamma$ is immediately equivalent to satisfiability of~$C$.
		\item $(\lnot)$: If $\Gamma$ matches $S, \lnot\phi \tin I$, then $M$ satisfies $\Gamma$ iff $\llbracket \lnot\phi \rrbracket_M \in I$ and for every $(\psi \tin J) \in S$ we have $\llbracket \psi \rrbracket_M \in J$. This is equivalent to $\llbracket \phi \rrbracket_M \in (1 - I)$ and $\llbracket \psi \rrbracket_M \in J$ for every $(\psi \tin J) \in S$, which are exactly the conditions for $C$ being satisfied.
		\item $(\ominus)$: If $\Gamma$ matches $S, \phi \ominus c \tin I$ with $I = \llparenthesis a,b \rrparenthesis$ and $0 \notin I$, then $M$ satisfies $\Gamma$ iff $\llbracket \phi \ominus c \rrbracket_M \in I$ and for every $(\psi \tin J) \in S$ we have $\llbracket \psi \rrbracket_M \in J$. The first condition is equivalent to $\llbracket \phi \rrbracket_M$ being in $(I+c)\cap[0,1]$. Again this is then directly equivalent to the conditions under which $C$ is satisfied.
		\item $(\ominus')$: If $\Gamma$ matches $S, \phi \ominus c \tin I$ with $I = \llparenthesis a,b \rrparenthesis$ and $0 \in I$, then $M$ satisfies $\Gamma$ iff $\llbracket \phi \ominus c \rrbracket_M \in I$ and for every $(\psi \tin J) \in S$ we have $\llbracket \psi \rrbracket_M \in J$. The first condition is equivalent to $\llbracket \phi \rrbracket_M \leq b+c$ or $\llbracket \phi \rrbracket_M < b+c$ depending on $\rrparenthesis$. This simplifies to $\llbracket \phi \rrbracket_M \in [0, b+c\rrparenthesis \cap [0,1]$.
		Again this is then directly equivalent to the conditions under which $C$ is satisfied.
		\item $(\land)$: If $\Gamma$ matches $S, \phi \land \psi \tin I$ with $I = \llparenthesis a,b \rrparenthesis$, then $M$ satisfies $\Gamma$ iff $\llbracket \phi \land \psi \rrbracket_M \in I$ and for every $(\psi \tin J) \in S$ we have $\llbracket \psi \rrbracket_M \in J$. Let $\triangleright = \geq$ if $\llparenthesis = [$ and $\triangleright = >$ otherwise. Similarly, let $\triangleleft = \leq$ if $\rrparenthesis = ]$ and $\triangleleft = <$ otherwise. Now $\llbracket \phi \land \psi \rrbracket_M$ being in $I$ is equivalent to both $\llbracket \phi \rrbracket_M \triangleright a$ and $\llbracket \psi \rrbracket_M \triangleright a$ being true and at least one of $\llbracket \phi \rrbracket_M \triangleleft b$ or $\llbracket \psi \rrbracket_M \triangleleft b$ being true. This then corresponds to $\llbracket \phi \rrbracket_M \in \llparenthesis a, b \rrparenthesis, \llbracket \psi \rrbracket_M \in \llparenthesis a, 1]$ or $\llbracket \phi \rrbracket_M \in \llparenthesis a, 1], \llbracket \psi \rrbracket_M \in \llparenthesis a, b \rrparenthesis$. Combined with the condition that for every $(\psi \tin J) \in S$ we have $\llbracket \psi \rrbracket_M \in J$, this is equivalent to the condition of either $C_1$ or $C_2$ being satisfied by $M$.
	\end{itemize}
\end{proof}

\begin{example}[Details of Example \ref{example:notEverythingExpBounded}]
	We show that the predicate lifting is natural. Let $h\colon X \rightarrow Y$, $f \in [0,1]^Y$, and $U \subseteq X$. Then we have:
	\begin{align*}
		\llbracket \heartsuit \rrbracket_X (f \circ h) (U) &= \sup_{x \in U, f(h(x)) \neq 1} f(h(x)) =\sup_{y \in \operatorname{Im} h \vert_U, f(y) \neq 1} f(y)\\
		&= \sup_{y \in \mathcal{P}(h)(U), f(y) \neq 1} f(y) =\llbracket \heartsuit \rrbracket_Y (f) \circ \mathcal{P}(h) (U)
	\end{align*}
\end{example}

\begin{remark}\label{rem:designated-witnesses}
  To prove soundness for a co-rule with $n>0$ conclusions, it suffices to construct a satisfying transition on the carrier $\{1,\ldots,n\}$ when state~$i$ realizes~$\Delta_i$. Indeed, given a pair $(X,\tau)$ realizing all the conclusions, choose a witness $f(i)\in X$ for each~$\Delta_i$ and give state~$i$ the valuation of~$f(i)$. If~$t$ is the transition on the indexed carrier, $Ff(t)$ gives a satisfying transition on~$X$ by naturality. The witnesses may realize several conclusions. For a rule with zero conclusions, it suffices to construct a transition on a singleton, since any nonempty carrier admits a map from the singleton and the conclusion has no successor constraints. In that case, the singleton construction must work for the valuation pulled back along that map.
\end{remark}

\begin{lemma}\label{localreduction2}
	Let $\mathcal{L}$ be a logic with a modal co-rule scheme. Then a tableau sequent $\Gamma$ over formulas $L \subseteq \mathcal{F}(\Lambda)$ with top-level decomposition $(V, \sigma_\Gamma, \Gamma^\sharp)$ is satisfiable if and only if there exists an open propositional tableau $G$ for $\Gamma^\sharp$ such that some modal co-rule prescribed by the modal co-rule scheme for $\Gamma^\sharp_G$ has conclusions $\Delta_1, \ldots, \Delta_s$ such that $\Delta_i\sigma_\Gamma$ is satisfiable for all $1 \leq i \leq s$.
\end{lemma}
\begin{proof}
  If~$\Gamma$ is satisfiable, \Cref{localreduction} gives a satisfying one-step model for~$\Gamma^\sharp$ in which all its successor valuations are jointly realized by the argument formulas. We apply the propositional rules in that model to obtain an open branch and use one-step tableau completeness on its saturated leaf. A state realizing~$\Delta_i$ then witnesses satisfiability of~$\Delta_i\sigma_\Gamma$ in the corresponding full model.

  Conversely, for each~$i$ choose a state satisfying~$\Delta_i\sigma_\Gamma$ in some model and let~$X$ index the chosen states. If there are no conclusions, choose the unique state of a singleton model and use a single index. Such a model exists because soundness of a zero-conclusion co-rule on a singleton implies $F1\ne\emptyset$. Define the valuation of each variable at an index by the value of its argument formula at the corresponding state. One-step soundness yields a transition satisfying the saturated leaf, so the propositional rules give satisfaction of~$\Gamma^\sharp$. The chosen states realize all the successor valuations, and \Cref{localreduction} yields satisfiability of~$\Gamma$.
\end{proof}

\begin{convention}
	For $n \in \mathbb{N}$, we will write $X_n := \{x_1, \ldots, x_n\}$ for a generic set containing $n$ elements.
\end{convention}

\begin{remark}
	In a polynomial-space bounded logic, given a premise~$\Gamma$ of size~$n$ with variables~$V$ and the exponential bound~$f_\mathcal{L}$ from~\Cref{def:exp-branching}, we can decide whether a set of exact tableau sequents $\{\Gamma_1, \ldots, \Gamma_{f_\mathcal{L}(n)}\}$ over $V$ is implied via a modal co-rule with a list of conclusions~$\Delta$ of the same length, that is, $\Gamma_j$ is a sub-sequent of $\Delta_j$ for all $j$, in $\PSPACE$. If such a co-rule exists, one-step soundness implies that for all $\tau$ with $\tau(v)(x_i) \in \Gamma_i(v)$ for all $v \in V, 1 \leq i \leq f_\mathcal{L}(n)$, there exists a $t \in F X_{f_\mathcal{L}(n)}$ such that $(X_{f_\mathcal{L}(n)}, \tau, t) \models \Gamma$. The test checks whether each $\Gamma_j$ is a sub-sequent of the corresponding~$\Delta_j$ for at least one modal co-rule with premise $\Gamma$.
\end{remark}

\begin{proof}[Proof of~\Cref{thm:mainCorrect}]
  We use induction over the modal depth, counting nullary modalities as having depth~$1$. If the modal depth is~$0$, \Cref{alg:sat} reduces to using the propositional tableau to decide satisfiability directly. Indeed, $\Gamma^\sharp_G=\emptyset$ for every open tableau~$G$, since its last label contains only $0$-literals. Since $F1\ne\emptyset$, there is a singleton model, so the empty sequent is satisfiable. Thus the algorithm is correct in this case by~\Cref{tableau:correctness}.

  For the induction step, we first represent~$\Gamma$ by its top-level decomposition $(V,\sigma_\Gamma,\Gamma^\sharp)$ and then non-deterministically choose an open propositional tableau~$G$ for~$\Gamma^\sharp$. If no open tableau exists, the algorithm cannot accept, and the tableau sequent is unsatisfiable by~\Cref{localreduction2}. If $\Gamma^\sharp_G=\emptyset$ after dropping the $0$-literals from the last label of~$G$, the algorithm accepts immediately. We use the transition and the values of the argument formulas at the unique state of a singleton model to obtain a one-step model. This model satisfies~$\Gamma^\sharp_G$ and hence~$\Gamma^\sharp$ by~\Cref{tableau:correctness}, so~$\Gamma$ is satisfiable by~\Cref{localreduction}.

  Now note that, for a chosen tableau~$G$ with $\Gamma^\sharp_G\ne\emptyset$, the condition in~\Cref{localreduction2} is that there is at least one modal co-rule for~$\Gamma^\sharp_G$ in~$\mathcal{R}$ such that all its conclusions are satisfiable (after resubstituting with~$\sigma_\Gamma$). We non-deterministically choose a modal co-rule and check each of its conclusions. A co-rule with zero conclusions passes this check immediately. Thus the algorithm reduces satisfiability to that of the resubstituted sequents, which we can recursively test. Each such sequent has lower modal depth, ensuring termination after linearly many nested calls. By the induction hypothesis, the checks for the chosen co-rule can all accept in a single run iff all its resubstituted conclusions are satisfiable. Together with the choice of~$G$, this gives the claimed equivalence.
\end{proof}

\begin{proof}[Proof of \Cref{lemma:generallybounds}]
	We begin with the case $\triangleright={>}$.
	Suppose that $\llbracket \GENERALLY v \rrbracket_M > a$.
	Then there exists $\alpha$ such that $a < \min(\alpha,t(\tau(v)_{\ge\alpha}))$.
	Then $\tau(v)_{\ge\alpha} \subseteq \tau(v)_{>a}$, so that $a < t(\tau(v)_{\ge\alpha}) \le t(\tau(v)_{>a})$.
	Conversely, if $\llbracket \GENERALLY v \rrbracket_M \le a$, then this means that $\min(\alpha, t(\tau(v)_{\ge\alpha})) \le a$ for each $\alpha\in[0,1]$, and hence, in particular, $t(\tau(v)_{\ge\alpha}) \le a$ for each $\alpha > a$.
	Therefore, $t(\tau(v)_{>a}) = t(\bigcup_{\alpha > a} \tau(v)_{\ge\alpha}) = \sup_{\alpha > a} t(\tau(v)_{\ge\alpha}) \le a$.

	We continue with the case $\triangleright={\ge}$.
	If $t(\tau(v)_{\ge a}) \ge a$, then $a = \min(a, t(\tau(v)_{\ge a})) \le \llbracket \GENERALLY v \rrbracket_M$ by definition.
	Conversely, suppose $t(\tau(v)_{\ge a}) < a$.
	Then we have $a > t(\tau(v)_{\ge a}) = t(\bigcap_{\alpha < a}(\tau(v)_{\ge \alpha})) = \inf_{\alpha < a} t(\tau(v)_{\ge\alpha})$, so there exists $b < a$ such that $t(\tau(v)_{\ge b}) < a$.
	Then:
	\begin{align*}
		\llbracket \GENERALLY v \rrbracket_M
		&= \sup_{\alpha\in[0,1]} \min(\alpha,t(\tau(v)_{\ge\alpha})) \\
		&= \max(\sup_{\alpha\le b} \min(\alpha,t(\tau(v)_{\ge\alpha})),
			\sup_{\alpha>b} \min(\alpha,t(\tau(v)_{\ge\alpha}))) \\
		&\le \max(\sup_{\alpha\le b} \alpha,
			\sup_{\alpha>b} t(\tau(v)_{\ge\alpha})) \\
		&\le \max(b, t(\tau(v)_{\ge b})) < a,
	\end{align*}
	where in the step to the last line we used that $\alpha>b$ implies that $\tau(v)_{\ge\alpha} \subseteq \tau(v)_{\ge b}$.

	The cases for upper bounds on $\llbracket \GENERALLY v \rrbracket_M$ can be reduced to the cases for lower bounds by negating both sides of the respective equivalences.
\end{proof}

\begin{proof}[Proof of \Cref{theorem:generallyexpbounded}]
	We can ignore modalities emulating atoms as in the proof of~\Cref{lemma:atoms}.
	Throughout, let ${\triangleright_1} = {>}$ if ${\llparenthesis_1} = {(}$, and ${\triangleright_1} = {\geq}$ otherwise. Furthermore, let ${\triangleleft_2} = {<}$ if ${\rrparenthesis_2} = {)}$, and ${\triangleleft_2} = {\leq}$ otherwise.
	We consider what it means for $\GENERALLY v$ to be in $\llparenthesis_1 a,b \rrparenthesis_2$ in a one-step model $M=(X, \tau, t)$:
	By \Cref{lemma:generallybounds}, the lower bound is satisfied iff $t(\{x \in X \mid \tau (v) (x) \triangleright_1 a\}) \triangleright_1 a$.
	Satisfaction of the upper bound is equivalent to $t(\{x \in X \mid 1 -\tau(v) (x) \triangleleft_2 ^\circ 1-b\}) \triangleleft_2 ^\circ 1-b$ being true. We then write $v_{\triangleright_1 a} := t(\{x \in X \mid \tau(v) (x) \triangleright_1 a\})$ and $(\lnot v)_{\triangleleft_2 ^\circ 1-b} := t(\{x \in X \mid 1-\tau(v) (x) \triangleleft_2 ^\circ 1-b\})$.
	Thus, if $\Gamma$ is a propositionally saturated one-step tableau sequent, then it is satisfiable iff there exists a one-step model $M=(X, \tau, t)$ such that for all $\GENERALLY v$ appearing in $\Gamma$ with $\Gamma(\GENERALLY v) = \llparenthesis_1 a,b \rrparenthesis_2$ we have $v_{\triangleright_1 a} \triangleright_1 a$ and $(\lnot v)_{\triangleleft_2 ^\circ 1-b} \triangleleft_2 ^\circ 1-b$.
	This allows us to reduce any one-step model to just its values for $v_{\triangleright_1 a}$ and $(\lnot v)_{\triangleleft_2 ^\circ 1-b}$ for all $\llparenthesis_1 a, b \rrparenthesis_2 = \Gamma(\GENERALLY v)$, and check for $v_{\triangleright_1 a} \triangleright_1 a$ and $(\lnot v)_{\triangleleft_2 ^\circ 1-b} \triangleleft_2 ^\circ 1-b$, giving us a reduced representation $r(M)$ of each one-step model $M$ as a vector in $[0,1]^{2n}$. We then have some one-step model $M$ satisfying all the conditions iff each coordinate of its vector representation $r(M)$ satisfies its respective inequality.
	
	Now, for each $x \in X$ of a one-step model $M=(X, \tau, t)$ we obtain a singleton one-step model $M_x = (\{x\}, \tau_x, t_x)$ where $\tau_x (v)(x) = \tau(v)(x)$ for all $v \in V$ and $t_x$ is the unique distribution on $\{x\}$. The reduced representation $r(M_x)$ of such a singleton one-step model then is a vector in $\{0,1\}^{2n}$, where a $1$ indicates that the bound of that respective position is met and a $0$ indicates that the bound is not met. 
	From now on, we write $\GENERALLY v_1, \ldots, \GENERALLY v_n$ for all one-step formulas appearing in $\Gamma$ and $(\llparenthesis_1)_i a_i, b_i (\rrparenthesis_2)_i = \Gamma(\GENERALLY v_i)$.
	We know that the vector representation $r(M)$ of $M$ is the convex combination of the vector representations $r(M_x)$ of $M_x$ with the coefficients of $t$:
	\begin{align*}
		r(M) &= \Bigg(\sum_{x \in X, \tau(v_1) (x) (\triangleright_1)_1 a_1} t(x), \sum_{x \in X, 1-\tau(v_1) (x) (\triangleleft_2)_1 ^\circ 1-b_1} t(x), \ldots,\\
		&\sum_{x \in X, \tau(v_n) (x) (\triangleright_1)_n a_n} t(x), \sum_{x \in X, 1-\tau(v_n) (x) (\triangleleft_2)_n ^\circ 1-b_n} t(x)\Bigg)^t\\
		&=\left(\sum_{x \in X} t(x) (r(M_x))_1, \sum_{x \in X} t(x) (r(M_x))_2, \ldots, \sum_{x \in X} t(x) (r(M_x))_{2n}\right)^t\\
		&= \sum_{x \in X} t(x) r(M_x)
	\end{align*}
	By Carathéodory's theorem, this implies that $r(M)$ can be written as a convex combination of at most $2n+1$ vector representations of the singleton one-step models. Choose distinct $x_1,\ldots,x_k\in X$, with $k\leq 2n+1$, and weights $\lambda_1,\ldots,\lambda_k\geq 0$ such that
	\begin{equation*}
		r(M) = \sum_{i=1}^k \lambda_i r(M_{x_i}), \qquad \sum_{i=1}^k \lambda_i = 1
	\end{equation*} 
	Then we obtain a new model $M' = (\{x_1, \ldots, x_k\}, \tau', t')$ by putting $\tau'(v)(x_i) = \tau(v)(x_i)$ for all $v \in V, 1 \leq i \leq k$ and putting $t'(x_i) = \lambda_i$. This model has the same vector representation, $r(M) = r(M')$, so $M$ satisfies $\Gamma$ iff $M'$ satisfies $\Gamma$.
	Thus, the existence of a one-step model that satisfies a tableau sequent $\Gamma$ is equivalent to the existence of a one-step model with at most $2n+1$ states that satisfies the same tableau sequent $\Gamma$.
\end{proof}

\begin{proof}[Proof of \Cref{theorem:generallyexpbranching}]
	Let $\Gamma$ be a propositionally saturated one-step tableau sequent with $n$ occurrences of formulas of the form $\GENERALLY v$. Again, we write $\GENERALLY v_1, \ldots, \GENERALLY v_n$ for the occurring formulas and $(\llparenthesis_1)_i a_i, b_i (\rrparenthesis_2)_i = \Gamma(\GENERALLY v_i)$.
	As apparent from the proof of \Cref{theorem:generallyexpbounded}, there are $2^{2n}$ possibilities for each of the $2n+1$ successor states for a one-step model satisfying $\Gamma$, so there are only exponentially many possible configurations of successor states; we restrict to configurations whose convex hull contains a vector satisfying the inequalities, i.e. we have a vector $r = \sum_{x \in X} t(x) r(M_x)$ with $t(x) \geq 0, \sum_{x \in X} t(x) = 1$ and $r$ satisfies $\Gamma$ in the following sense: $r \in [0,1]^{2n}$ satisfies $\Gamma$ if for all $1 \leq i \leq n$ we have $r_{2i-1} (\triangleright_1)_i a_i$ and $r_{2i} (\triangleleft_2)_i ^\circ 1 - b_i$. Each such state then has to either have $\tau(v_i)(x) (\triangleright_1)_i a_i$ or $\tau(v_i)(x) \overline{(\triangleright_1)_i}^\circ a_i$ for each of the odd-numbered $n$ inequalities and $1-\tau(v_i)(x) (\triangleleft_2)_i^\circ 1-b_i$ or $1-\tau(v_i)(x) \overline{(\triangleleft_2)_i} 1 - b_i$ for the other $n$ inequalities.
	Taking the inequalities for each variable~$v$ at each state~$x_i$ and writing them as an interval $\Delta_i(v)$ then gives us a list of exact tableau sequents $(\Delta_1, \ldots, \Delta_{2n+1})$ by construction, ordered by their respective binary representation. These tableau sequents are the conclusions of one modal co-rule. Taking all exponentially many possible configurations that can satisfy the inequalities and generating a modal co-rule for each in this way then gives us an exponentially bounded set of modal co-rules $R_1, \ldots, R_m$ with $m \leq 2^{2n \cdot (2n+1)}$ by construction. More explicitly, we put $\matches_{\mathcal{R}}(\Gamma) = 2^{2n \cdot (2n+1)}$ and $\concls_{\mathcal{R}}(\Gamma, i) = 2n+1$ if the $i$-th co-rule exists and $\bot$ otherwise.
	We can once again ignore modalities emulating atoms by \Cref{lemma:atoms}.
\end{proof}

\begin{remark}
	Intuitively, we can describe the set of modal co-rules one constructs in the proof of \Cref{theorem:generallyexpbranching} in the following way: We describe each possible successor by whether it will count towards satisfying $\GENERALLY v \triangleright a$ and/or towards satisfying $\GENERALLY v \triangleleft b$ for each $\GENERALLY v \tin \llparenthesis a, b \rrparenthesis$. This immediately tells us for each possible successor what interval each $v$ must be in and allows us to display each successor as a truth vector detailing which bounds it counts towards. Putting all these intervals for $v$ together for one possible assignment of states and which bounds they count towards then gives us a list of exact tableau sequents $(\Delta_1, \ldots, \Delta_{2n+1})$ over $V$. We now filter for configurations where we can find a~$t$ that assigns weights in such a way that all literals $\GENERALLY v \tin \llparenthesis a, b \rrparenthesis$ are satisfied in the sense that to states where $v \triangleright a$ holds we assign a combined value $t_v \triangleright a$ and to states where $v \triangleleft b$ holds a combined value of $t_{\lnot v} \triangleleft^\circ 1-b$. This will ensure that, as long as we choose $\tau(v)(x_i) \in \Delta_i(v)$ for each $x_i \in X_{2n+1}$ and $v \in V$, we obtain a one-step model satisfying $\Gamma$ by using this~$t$.
	These configurations then correspond to ones where the convex hull of the truth vector description of states of the configuration contains an element that has values satisfying the inequalities outlined for $t$.
	Taking all possible configurations, filtering out only those that can result in a $t$ satisfying the inequalities, and then describing this configuration as a set of exact tableau sequents over $V$, as outlined, yields a set of modal co-rules by construction.
\end{remark}

\begin{example}[Expected value]\label{ex:expectation-corules}
  Take $F=\Dist$ and $\Lambda=\{\PROBABLY\}$, with the expected-value interpretation above. Any one-step model $(X,\tau,t)$ can be replaced by a model on a singleton~$\{*\}$ with its Dirac distribution and $\tau'(v)(*)=\sum_{x\in X}t(x)\tau(v)(x)$ for every variable~$v$. This preserves the value of every one-step formula.

  To exhibit a premise for which no finite set of sound co-rules is complete, take two distinct variables~$v,w$, which may come from two independent atoms in the full logic, and put $\Gamma=\{\PROBABLY v\tin[1/2,1/2],\,\PROBABLY w\tin[1/2,1/2]\}$. On a carrier $\{x,y\}$, assign the pairs of values for $(v,w)$ as $(0,r)$ at~$x$ and $(1,1-s)$ at~$y$, where $r,s\in(0,1/2)$. The constraint on~$\PROBABLY v$ forces both weights to be~$1/2$. With those weights, the value of~$\PROBABLY w$ is $(1+r-s)/2$, so the pair extends to a model satisfying~$\Gamma$ exactly when $r=s$.

  Suppose a finite set of sound co-rules is complete for~$\Gamma$. Only finitely many conclusion sequents occur in this set. As~$r$ varies, the valuation $(0,r)$ at~$x$ determines which of those conclusions hold at that state. Choose distinct~$r,r'$ for which these sets of conclusions agree. The carrier with valuations $(0,r)$ and $(1,1-r)$ satisfies~$\Gamma$ with equal weights, so completeness gives a matching co-rule whose conclusions are all realized. Replacing the valuation of~$x$ by $(0,r')$ leaves all those conclusions realized. Soundness would then give a distribution satisfying~$\Gamma$ on the new pair, which is impossible since $r'\ne r$. Thus the logic fails exponential branching even though one state suffices for unrestricted one-step satisfiability.
\end{example}

\begin{example}
	\begin{enumerate}
		\item Fix $c\in[0,1]\cap\Rat$. We prove that the tableau sequent $\Gamma := \{\GENERALLY p \tin [0,c], \GENERALLY q \tin [0,c], \GENERALLY (p \land q) \tin (c,1]\}$ is not satisfiable, which means that in states satisfying both $\GENERALLY p \tin [0,c]$ and $\GENERALLY q \tin [0,c]$, we also have $\GENERALLY (p \land q) \tin [0,c]$. Doing top-level decomposition gives us $\GENERALLY v_1 \tin [0,c], \GENERALLY v_2 \tin [0,c], \GENERALLY v_3 \tin (c,1]$. Then the tableau sequents for the possible successor states (filtering out immediately unsatisfiable ones) would be:
		\begin{align*}
			&q_1 := \{v_1 \tin [0,c], v_2 \tin [0,c], v_3 \tin (c,1]\}
			&&q_2 := \{v_1 \tin (c,1], v_2 \tin [0,c], v_3 \tin (c,1]\}\\
			&q_3 := \{v_1 \tin [0,c], v_2 \tin (c,1], v_3 \tin (c,1]\}
			&&q_4 := \{v_1 \tin [0,c], v_2 \tin [0,c], v_3 \tin [0,c]\}\\
			&q_5 := \{v_1 \tin (c,1], v_2 \tin (c,1], v_3 \tin (c,1]\}
			&&q_6 := \{v_1 \tin (c,1], v_2 \tin [0,c], v_3 \tin [0,c]\}\\
			&q_7 := \{v_1 \tin [0,c], v_2 \tin (c,1], v_3 \tin [0,c]\}
			&&q_8 := \{v_1 \tin (c,1], v_2 \tin (c,1], v_3 \tin [0,c]\}
		\end{align*}
		The conclusions of possible modal co-rules are then all combinations of $7$ such elements, where we can choose a distribution such that it yields a one-step model satisfying the sequent.
		Taking a $q_i$ and substituting the corresponding formulas for the variables gives us the tableau sequents $q'_i$.
		Now we can show that $q'_1, q'_2, q'_3, q'_8$ are not satisfiable. We show this explicitly for $q'_1$:
		{
			\begin{gather*}
				\hspace*{10mm}
				\cfrac{
					\begin{array}{l}
						p \tin [0, c], q \tin [0,c], p \land q \tin (c,1]
					\end{array}
				}{}(\land)\\
				\cfrac{
					\begin{array}{l}
						p \tin [0, c], q \tin [0,c],\\ p \tin (c,1], q \tin (c,1]
					\end{array}
				}{}(\cap)
				\hspace*{6mm}
				\cfrac{
					\begin{array}{l}
						p \tin [0, c], q \tin [0,c],\\ p \tin (c,1], q \tin (c,1]
					\end{array}
				}{}(\cap)\\			
				\cfrac{
					\begin{array}{l}
						p \tin \emptyset, q \tin \emptyset
					\end{array}
				}{\bot}(\text{Ax})
				\hspace*{20mm}
				\cfrac{
					\begin{array}{l}
						p \tin \emptyset, q \tin \emptyset
					\end{array}
				}{\bot}(\text{Ax})
			\end{gather*}
		}
		The same can be shown for $q'_2, q'_3, q'_8$, meaning that only combinations including $q_5$ can potentially yield a distribution satisfying $\GENERALLY v_3 \tin (c,1]$. However, this means we would have to assign a probability greater than $c$ to the state realizing $q_5$ in such a distribution. This then would also immediately imply that we have $\GENERALLY v_1 \tin (c,1], \GENERALLY v_2 \tin (c,1]$. As such, $\Gamma$ is not satisfiable.
		\item We now investigate $\Gamma := \{\GENERALLY p \tin [0,c], \GENERALLY q \tin [0,c], \GENERALLY (p \lor q) \tin (c,1]\}$. At $c=0$, the first two bounds force $p$ and $q$ to be zero almost surely by~\Cref{lemma:generallybounds}, so $\GENERALLY(p\lor q)$ also has value~$0$. At $c=1$, the interval $(c,1]$ is empty. Thus $\Gamma$ is unsatisfiable at both endpoints. Assume $0<c<1$. After top-level decomposition, we obtain the same tableau sequents $q_1, \ldots, q_8$ and co-rules as in the previous example, with the only difference that we now have that $v_3$ will be substituted with $(p \lor q)$ instead of $(p \land q)$. This yields us $q'_1, \ldots, q'_8$ by resubstituting the formulas for their variables. One easily spots that $q'_1, q'_6, q'_7, q'_8$ are unsatisfiable. On the other hand, we immediately see that $q'_2, \ldots, q'_5$ are satisfiable. We can thus reduce our search to distributions over $4$ states $x_1, \ldots, x_4$, where $x_i$ realizes $q'_{i+1}$. We can assign at most probability $c$ to the state pairs $(x_1, x_4)$ and $(x_2, x_4)$ to ensure the resulting model will still satisfy the literals $\GENERALLY v_1 \tin [0,c]$ and $\GENERALLY v_2 \tin [0,c]$ respectively. At the same time, we have to assign more than probability $c$ to the state triple $(x_1, x_2, x_4)$ to satisfy $\GENERALLY v_3 \tin (c,1]$. Taking $0<\epsilon<\min(c,1-c)/2$, we can set the probability to $\frac{c}{2} + \epsilon$ for each of the states $x_1$ and~$x_2$, to~$0$ for~$x_4$, and to $1 - c - 2 \epsilon$ for~$x_3$. This then yields a distribution~$t$ such that we obtain the one-step model $M=(X:= \{x_1, x_2, x_3, x_4\}, \tau, t)$, where:
		\begin{gather*}
			\tau(v_1)( x_i) := \begin{cases*}
				0, &\text{if }{i}={2,3}\\
				\frac{1+c}{2}, &\text{else}
			\end{cases*}
			\qquad
			\tau(v_2)( x_i) := \begin{cases*}
				0, &\text{if }{i}={1,3}\\
				\frac{1+c}{2}, &\text{else}
			\end{cases*}\\
			\tau(v_3)( x_i) := \begin{cases*}
				0, &\text{if }{i}={3}\\
				\frac{1+c}{2}, &\text{else}
			\end{cases*}
		\end{gather*}
	\end{enumerate}
\end{example}

\begin{proof}[Proof of~\Cref{thm:main}]
  If $F1=\emptyset$, all sequents are unsatisfiable. Otherwise, we once again use induction over the modal depth. If the modal depth is~$0$, the tableau algorithm to decide propositional satisfiability works in nondeterministic polynomial time.

  For the induction step, a top-level decomposition of~$\Gamma$ can be computed in polynomial time, and guessing a propositional tableau~$G$ for~$\Gamma^\sharp$ can be done in nondeterministic polynomial time in the size of the current sequent. Each propositional step replaces its principal formulas by smaller formulas or merges intervals. Since $\matches_{\mathcal{R}}$ and $\concls_{\mathcal{R}}$ are both bounded by exponentials, the indices we choose have at most polynomial length in binary, and we can encode a specific co-rule $R\in\mathcal{R}_{\Gamma^\sharp_G}$ in polynomial space in the size of its premise.

  To bound the sizes of the recursive sequents, let~$N$ be the size of the initial input and let~$D$ be the least common multiple of the reduced denominators of its interval endpoints and shift constants. The binary length of~$D$ is at most the sum of the binary lengths of those denominators, up to an additive constant. All endpoints initially lie on the grid $\{k/D\mid k=0,\ldots,D\}$. Reflecting endpoints, shifting them by input constants, and intersecting intervals preserve this grid. In a co-rule application, the additive subgroup generated by~$1$ and the endpoints of the premise is contained in $\frac{1}{D}\mathbb{Z}$, so all endpoints of the conclusions stay on the same grid. Thus each endpoint has polynomial binary length in~$N$. Each co-rule conclusion has at most one interval per variable, and the resubstituted sequents contain only subformulas of the initial input, with intervals for repeated formulas intersected. Therefore both the co-rule conclusions and the recursive sequents have size polynomial in~$N$.

  Computing the number of matching co-rules, the number of conclusions, and any individual conclusion $\concl_{\mathcal{R}}(\Gamma^\sharp_G,i,j)$ can be done in polynomial space by polynomial-space boundedness, since the premise has size polynomial in~$N$. Resubstituting with~$\sigma_\Gamma$ also uses polynomial space in~$N$. Satisfiability of $\concl_{\mathcal{R}}(\Gamma^\sharp_G,i,j)\sigma_\Gamma$ can then be checked recursively, and the space needed for the recursive call is reused after each loop iteration.

  Since the modal depth is bounded by the syntactic size of the initial sequent and decreases in each recursive call by~\Cref{thm:mainCorrect}, there are at most $N+1$ calls on the recursion stack. Each call uses at most polynomial space in~$N$ apart from its recursive calls, so in total, the whole algorithm uses at most polynomial space. The algorithm is nondeterministic, so Savitch's theorem~\cite{Savitch70} gives the deterministic polynomial-space bound.
\end{proof}

\begin{remark}\label{remark:freshvs}
	We could have also defined these properties only for tableau sequents over one-step formulas where each $v \in V$ is used in at most one tableau literal, and still would have obtained that any logic that has these properties has its satisfiability problem in $\PSPACE$. The reason for this is that when doing a top-level decomposition of some $\Gamma$, we already introduce a fresh $v \in V$ for every modal operator, regardless of whether the formula that is being substituted has been seen before.
\end{remark}

\begin{proof}[Proof of~\Cref{lemma:atoms}]
  For exponential branching, the predicate liftings of modalities emulating atoms do not depend on any successor valuations. Let~$\Gamma^{-}$ be obtained from a premise~$\Gamma$ by removing the literals for these modalities. As long as the intervals for each atom have nonempty intersection, the modal co-rules of the original logic still define the modal co-rule scheme of the new logic. We take the co-rules for~$\Gamma^{-}$ and replace their premises by~$\Gamma$. Since atoms are nullary, the same variables occur in both premises, so we can use the conclusions directly. Soundness follows by extending the transition with a value from the intersection for each occurring atom and arbitrary values for the remaining atoms, and completeness follows by forgetting the atom values in a model of~$\Gamma$ and applying completeness of the original scheme. If the intervals for some atom have empty intersection, no model satisfies~$\Gamma$ and we use no co-rules. The bounds on the numbers of co-rules and conclusions carry over.

  For polynomial-space boundedness, we can check in polynomial space whether some modalities emulating the same atom have conflicting bounds, in which case the modal co-rule scheme returns no co-rules. Otherwise, computing the number of co-rules, the number of conclusions, and any individual conclusion reduces to the corresponding computations for~$\Gamma^{-}$. Every endpoint of $\Gamma^{-}$ is an endpoint of~$\Gamma$, so the subgroup generated by~$1$ and the endpoints of $\Gamma^{-}$ is contained in the corresponding subgroup for~$\Gamma$.
\end{proof}

For the metric modalities, the finite one-step models used in~\Cref{sec:case-studies} must also have only finitely many positive labelled transitions.
An existing finite model property for sequents with non-strict bounds also covers mixed sequents, since
each strict bound can be tightened to a non-strict rational bound still satisfied at the given state.
In~\Cref{lem:sup-finite-models,lem:quantitative-finite-models}, we include limits of realizable tuples in the induction
to account for suprema that need not be attained.

For the metric modalities below, write $h_{\ell,c}(m)=c\ominus d_L(\ell,m)$.

\begin{lemma}\label{lem:sup-finite-models}
  For non-expansive fuzzy~$\ALC$ over the empty TBox and for metric modal logic over a proper label space, the truth values of any finite set of formulas at a state in a model with arbitrary transition support can be realized at a state in a finite model with finite transition support. This also holds for metric modal logic with crisp transitions.
\end{lemma}
\begin{proof}
  Fix a finite set of formulas closed under subformulas. The \emph{truth-value profile} of a state for a collection of formulas is the tuple of their truth values at that state. For each depth~$k$, consider the closure in the finite-dimensional unit cube of the profiles of formulas of depth at most~$k$ realized at states in arbitrary models. Here propositional atoms have depth~$0$. We show by induction on~$k$ that each such limit profile is realized in a finite model with finite transition support. At depth~$0$, choose the atom values of the profile and evaluate the propositional operations, which are continuous.

  At depth~$k>0$, choose a sequence of models~$M_n$ and states~$x_n$ in them whose profiles converge to the desired profile.
  For fuzzy~$\ALC$, let~$S$ be the closure of the jointly realized tuples of values of the modal arguments.
  By induction, each tuple in~$S$ is realized in a finite model. The admissible tuples~$(\min(r,a_i))_i$
  used in~\Cref{sec:case-studies} form a compact set as the continuous image of~$[0,1]\times S$.
  Apply the one-step finite model property for crisp roles at~$x_n$. For each positive target value, the corresponding
  diamond value at~$x_n$ is positive for large~$n$ and hence attained in the finite crisp model.
  A preimage of the witnessing tuple gives an edge degree and a jointly realizable successor valuation.

  For metric modal logic, choose for each positive target value an edge from~$x_n$ in each~$M_n$ whose contribution approaches that target as the sequence progresses. Its contribution is the minimum of the edge degree, the argument value, and the label cap~$h_{\ell,c}(m)$ of the edge label~$m$. For each of the finitely many selected edge sequences in either logic, pass to a subsequence along which the edge degree and the joint profile of the successor formulas of depth at most~$k-1$ converge. In the metric case, sufficiently late edges in the sequence have positive contribution, so their labels lie in the closed ball of radius~$1$ around the label of the modality. Properness lets us also take a convergent subsequence of those labels. For crisp transitions, each selected edge has degree~$1$.

  Use the induction hypothesis to realize each limit successor profile at a state in a separate finite model. Add a fresh root with the limiting atom values and one edge to each chosen state, using the limiting degrees and, in the metric case, the limiting labels. All other edges from the root have degree~$0$. The selected edge for each positive target attains that target by continuity. For each~$n$, the contribution of a selected edge to any other modal subformula is at most the value of that subformula at~$x_n$. Passing to the limit gives the corresponding inequality for the target values. This also preserves modal targets equal to~$0$. The propositional formulas have the required values by continuity.
\end{proof}

\begin{lemma}
  In non-expansive fuzzy~$\ALC$ over~$\fPow^{\mathrm{fin}}$, every satisfiable propositionally saturated one-step tableau sequent with~$n$ modal formulas has a one-step model with at most~$\max(1,n)$ states.
\end{lemma}
\begin{proof}
  Given a satisfiable propositionally saturated one-step tableau sequent~$\Gamma$ with~$n$ modal formulas, we construct a one-step model with at most $\max(1,n)$ states. If there are no modal formulas, use a singleton with the zero transition. Otherwise, write $\Diamond v_1,\ldots,\Diamond v_n$ for the modal formulas in~$\Gamma$ and restrict attention to their variables. We take a one-step model $M=(X,\tau,t)$ satisfying~$\Gamma$ and define $M'=(\{x_1,\ldots,x_n\},\tau',t')$ by $t'(x_i)=\llbracket\Diamond v_i\rrbracket_M$, $\tau'(v_i)(x_i)=\llbracket\Diamond v_i\rrbracket_M$, and $\tau'(v_i)(x_j)=0$ for $i\ne j$. The model~$M'$ satisfies~$\Gamma$, as $\llbracket\Diamond v_i\rrbracket_{M'}=\llbracket\Diamond v_i\rrbracket_M$ for all~$i$.
\end{proof}

\begin{lemma}\label{lem:fuzzy-alc-branching}
  Non-expansive fuzzy~$\ALC$ over~$\fPow^{\mathrm{fin}}$ is exponentially branching.
\end{lemma}
\begin{proof}
  Given a propositionally saturated one-step tableau sequent~$\Gamma$ with~$n$ modal formulas, we build a modal co-rule with at most~$n$ exact tableau sequents as conclusions. Write $\Diamond v_1,\ldots,\Diamond v_n$ for these formulas. If any premise interval is empty, we use no co-rules. Otherwise, let $\Gamma(\Diamond v_i)=\llparenthesis_i a_i,b_i\rrparenthesis_i$ for $1\le i\le n$, with lower interval $A_i=\llparenthesis_i a_i,1]$ and upper interval $B_i=[0,b_i\rrparenthesis_i$. For each~$i$ with $A_i\ne[0,1]$, put $\Delta_i(v_i)=A_i$. For $j\ne i$, put $\Delta_i(v_j)=B_j$ if $A_i\cap B_j=\emptyset$, and $\Delta_i(v_j)=[0,1]$ otherwise. If there are no nontrivial lower bounds, the co-rule has zero conclusions.

  To prove soundness, work on a carrier with one state~$x_i$ per conclusion, realizing~$\Delta_i$, as in~\Cref{rem:designated-witnesses}. We can find a transition~$t$ satisfying~$\Gamma$ by taking $t(x_i)$ to be a small enough value in~$\Delta_i(v_i)$. More specifically, take a value from $A_i\cap\bigcap_{j:\,A_i\cap B_j\ne\emptyset}B_j$. The downward intervals~$B_j$ are nested, so this intersection is nonempty. Every upper bound is then met either by the edge degree or by the value of the corresponding variable required by~$\Delta_i$. The lower bound on~$\Diamond v_i$ is met at~$x_i$. If there are zero conclusions, the zero transition suffices.

  For completeness, let $(X,\tau,t)$ satisfy~$\Gamma$. Since the support is finite, for each nontrivial lower bound there is a state~$x$ such that both $t(x)$ and $\tau(v_i)(x)$ belong to~$A_i$. For every~$j$ with $A_i\cap B_j=\emptyset$, the edge degree $t(x)$ lies outside~$B_j$. If $\tau(v_j)(x)$ also lay outside~$B_j$, the upper bound on~$\Diamond v_j$ would fail. Hence~$x$ realizes~$\Delta_i$. This gives a sound and complete scheme with one co-rule for each premise with nonempty intervals and at most~$n$ conclusions.
\end{proof}

\begin{example}
	We give an instance of the modal co-rule in non-expansive fuzzy $\ALC$: We interpret the modality $\Diamond$ as the degree of thermal coupling between components and take $\mathsf{highTemp}, \mathsf{highLoad}, \mathsf{poorCooling}$ as atoms.
	We investigate the tableau sequent $\Gamma = \{\mathsf{highTemp} \land \Diamond (\mathsf{highLoad} \lor \mathsf{poorCooling}) \tin [0.7,0.9], \Diamond \mathsf{highTemp} \tin [0, 0.5]\}$: It being satisfiable would mean there can be components that have a high temperature and affect another component, either with poor cooling or a high workload, in a critical way, but not enough to cause their temperature to rise to a dangerous degree.
  The top-level decomposition gives $V=\{v_1,v_2\}$, $\sigma_\Gamma(v_1)=\mathsf{highLoad}\lor\mathsf{poorCooling}$, and $\sigma_\Gamma(v_2)=\mathsf{highTemp}$. After applying the conjunction rule, choose the branch with $\Diamond v_1\tin[0.7,0.9]$ and $\mathsf{highTemp}\tin[0.7,1]$. The latter constraint is handled by~\Cref{lemma:atoms}. For the remaining modal premise we obtain the co-rule
	\begin{equation*}
		\frac{\Diamond v_1 \tin [0.7,0.9], \Diamond v_2 \tin [0, 0.5]}{\{v_1 \tin [0.7,1], v_2 \tin [0, 0.5]\}}
	\end{equation*}
  The lower bound~$0$ on $\Diamond v_2$ needs no witness. The witness for $\Diamond v_1$ requires an edge degree of at least~$0.7$, so the upper bound~$0.5$ on $\Diamond v_2$ has to be met by the value of~$v_2$ at that witness. The upper bound~$0.9$ on $\Diamond v_1$ can be met by choosing the edge degree to be~$0.7$.
\end{example}

\begin{lemma}\label{lem:fuzzy-alc-polyspace}
  Non-expansive fuzzy~$\ALC$ over~$\fPow^{\mathrm{fin}}$ is polynomial-space bounded.
\end{lemma}
\begin{proof}
  The procedure outlined in~\Cref{lem:fuzzy-alc-branching} already provides an algorithm to compute the number of co-rules, the number of conclusions, and any individual conclusion in polynomial space. Each endpoint in a conclusion is either $0$, $1$, or an endpoint of the premise, so the denominator condition holds.
\end{proof}

\begin{corollary}
	The satisfiability problem of non-expansive fuzzy $\ALC$ is in $\PSPACE$.
\end{corollary}

\begin{lemma}\label{lem:quantitative-finite-models}
  For non-expansive quantitative fuzzy~$\ALC$, the truth values of any finite set of formulas at a state in a $\Dist$-model can be realized at a state in a finite $\Dist_{\mathrm{fin}}$-model.
\end{lemma}
\begin{proof}
  Fix a finite set of formulas closed under subformulas. As in \Cref{lem:sup-finite-models}, we show by induction on depth that every limit of realized profiles is realized in a finite model with finite transition support. At depth~$0$, choose the limiting atom values and use continuity of the propositional operations.

  At positive depth, choose a sequence of models~$M_n$ and states~$x_n$ in them whose profiles converge to the target. Write~$a_i$ for the target value of the modal subformula~$\boldsymbol{M}_{p_i}\phi_i$. The modality with $p_i=1$ is constant~$0$, so consider only~$p_i<1$.
  At~$x_n$, impose strict rational bounds on the modal values that allow each to vary by less than~$1/n$ in the unit interval.
  By the local finite model property for strict bounds proved by Schr\"oder and Pattinson~\cite[Example~4.6]{SchroderPattinson11}, these bounds have a finite one-step model in which each successor tuple of lower-depth formula values is realized in~$M_n$,
  since we may take precisely these tuples as the admissible successor valuations. The modal values of the finite one-step models still converge to the~$a_i$.
  Choose positive numbers $\varepsilon_n\to0$ larger than the deviations of these values from their targets. For $a_i>0$, record membership in the cut $\phi_i\ge a_i-\varepsilon_n$. For $a_i<1$, also record membership in the cut $\phi_i>a_i+\varepsilon_n$. For sufficiently large~$n$, the cuts are in the unit interval. In the finite one-step model the first cut has mass strictly above~$p_i$ and the second has mass at most~$p_i$, since the modal value is the largest supported argument value whose upper tail has mass greater than~$p_i$.

  The finitely many cuts give finitely many Boolean patterns. Pass to a subsequence on which the set of positive-mass patterns is fixed, and for each such pattern choose a representative successor in each model. Pass to a further subsequence along which the representatives' joint lower-depth profiles converge. Keep the probability masses of the patterns from one model in this subsequence. In the limit, a member of a lower cut has value at least~$a_i$, and a nonmember of an upper cut has value at most~$a_i$. Thus, with the fixed masses, $\Pr(\phi_i\ge a_i)>p_i$ when $a_i>0$, and $\Pr(\phi_i>a_i)\le p_i$ when $a_i<1$. Keeping the masses fixed preserves the strict inequalities.

  Use the induction hypothesis to realize each limit profile at a state in a separate finite model. Attach a fresh root with the chosen distribution over these states and with the limiting atom values. The two inequalities give modal value~$a_i$: the lower inequality makes~$a_i$ an admissible threshold, and the upper inequality excludes every larger threshold. The lower bound at~$0$ and the upper bound at~$1$ hold automatically. This realizes the target profile, including the propositional formulas by continuity.
\end{proof}

\begin{lemma}\label{quantitativeexpbranching}
  Non-expansive quantitative fuzzy~$\ALC$ over~$\Dist_{\mathrm{fin}}$ is exponentially branching.
\end{lemma}
\begin{proof}
  We first investigate what it means for $\boldsymbol{M}_p v$ to have a value in~$\llparenthesis a,b\rrparenthesis$ in a one-step model $M=(X,\tau,t)$ with finite support. Write $v_{\bowtie c}:=t(\{x\in X\mid\tau(v)(x)\bowtie c\})$. For $p<1$, the value of~$\boldsymbol{M}_p v$ is the largest supported truth value whose upper tail has mass greater than~$p$. The lower bound $\llbracket\boldsymbol{M}_p v\rrbracket_M\triangleright a$ is therefore equivalent to $v_{\triangleright a}>p$, and the upper bound $\llbracket\boldsymbol{M}_p v\rrbracket_M\triangleleft b$ is equivalent to $v_{\triangleleft b}\ge1-p$, with the comparisons chosen according to the brackets. A closed lower bound at~$0$ is vacuous, and a strict upper bound at~$0$ is impossible. For $p=1$, check the interval against the constant value~$0$.

  This allows us to represent each one-step model by its values for $v_{\triangleright a}$ and $v_{\triangleleft b}$ for the modal literals in a premise~$\Gamma$. With~$n$ literals, this gives a vector of at most~$2n$ probabilities. Satisfaction of~$\Gamma$ is determined by the corresponding inequalities and the constant tests.

  This vector is a convex combination of the Boolean vectors of the successor states, where a~$1$ at a position means the state belongs to the corresponding threshold set and a~$0$ means it lies outside that set. By Carathéodory's theorem, the same vector can be written as a convex combination of at most $2n+1$ of the vectors that occur in the support. Keeping a state for each selected vector and giving it the corresponding weight preserves all the bounds of~$\Gamma$.

  There are thus at most~$2^{2n}$ possibilities for each of at most~$2n+1$ successors, so there are only exponentially many possible configurations. We restrict to nonempty configurations whose convex hull contains a vector satisfying the probability inequalities of the premise. If the premise fails a constant test, we use no co-rules for it.

  Each Boolean vector requires the value of a variable to satisfy or to violate the threshold test for each of its coordinates. Taking the inequalities for each variable~$v$ at each state~$x_i$ and intersecting them as an interval for~$\Delta_i(v)$ gives an exact tableau sequent~$\Delta_i$ over~$V_\Gamma$, using~$[0,1]$ for unspecified variables. Discard any configuration that gives an empty interval. We take these sequents as the conclusions of a modal co-rule and enumerate the co-rules by the binary representations of their configurations.

  Completeness follows since the selected vectors are realized in the original finite-support model. For soundness, use the weights of a feasible convex combination as the distribution on the designated witnesses of~\Cref{rem:designated-witnesses}. The resulting probabilities satisfy the required inequalities, so the one-step model satisfies~$\Gamma$.
\end{proof}

\begin{theorem}\label{quantitativepolynomialspace}
  Non-expansive quantitative fuzzy~$\ALC$ over~$\Dist_{\mathrm{fin}}$ is polynomial-space bounded.
\end{theorem}
\begin{proof}
  Following the proof of~\Cref{quantitativeexpbranching}, we compute the $j$-th conclusion of the $i$-th modal co-rule by iterating over the possible configurations of successors and checking which ones satisfy the inequalities corresponding to the bounds of the premise~$\Gamma$. We take the $i$-th feasible configuration, then take the vector of its $j$-th successor and construct the corresponding tableau sequent. Feasibility for a configuration is a rational linear programming problem. To handle strict inequalities, use a common slack variable bounded by~$1$ and check whether its maximum is positive. The configurations and linear programs have polynomial size, so counting the feasible configurations and computing their conclusions uses polynomial space. Each endpoint in a conclusion is either $0$, $1$, or an endpoint of the premise.
\end{proof}

\begin{lemma}\label{metricexpbounded}
  For metric modal logic with either quantitative or crisp transitions over the corresponding finite-support functor, every satisfiable propositionally saturated one-step tableau sequent with~$n$ modal literals has a one-step model with at most~$\max(1,n)$ states.
\end{lemma}
\begin{proof}
  For each nontrivial lower bound, retain one labelled edge that witnesses it in the finite support. Deleting the other edges preserves all upper bounds. Restrict the carrier to the retained successor states, using a singleton when there are none. This gives at most one state per literal, or one for an empty premise.
\end{proof}

\begin{lemma}\label{metricexpbranching}
  Under the polynomial-space ball-intersection assumption, metric modal logic with either quantitative or crisp transitions is exponentially branching over the corresponding finite-support functor.
\end{lemma}
\begin{proof}
  We first investigate what $\Diamond_{\ell_i}^{c_i}v_i\tin\llparenthesis_i a_i,b_i\rrparenthesis_i$ means for a one-step model $M=(X,\tau,t)$ with finite transition support. Write $A_i=\llparenthesis_i a_i,1]$ and $B_i=[0,b_i\rrparenthesis_i$ for the lower and upper intervals. If any premise interval is empty, we use no co-rules. To satisfy a nontrivial lower bound, there must be a labelled successor $(m,x)$ such that $\tau(v_i)(x)$, $t(m,x)$, and $h_{\ell_i,c_i}(m)$ all belong to~$A_i$. To satisfy an upper bound, at least one of these three quantities must belong to~$B_i$ at every labelled successor in the support.

  We now construct a set of modal co-rules. Let~$\Gamma$ be a premise with~$n$ modal literals, written as above. We define the exact conclusions~$\Delta_i$ over~$V_\Gamma$ of a co-rule, one for each nontrivial lower bound, by putting $\Delta_i(v_i)=A_i$. Let~$E=[0,1]$ for quantitative transitions and $E=\{1\}$ for crisp transitions, and write $E_i=A_i\cap E$ for the degrees available to a witnessing edge. For each upper bound~$B_j$ with $E_i\cap B_j=\emptyset$, we either require $v_j\in B_j$ in~$\Delta_i$ or leave that upper bound to the label of the witnessing edge. We intersect the intervals for repeated variables and use~$[0,1]$ for variables with no requirements. Discard choices with inconsistent intervals or for which there is no suitable label~$m_i$ for some witness~$i$. A suitable label satisfies $h_{\ell_i,c_i}(m_i)\in A_i$ and $h_{\ell_j,c_j}(m_i)\in B_j$ for every upper bound assigned to the label. These tests reduce to checking intersections of open or closed balls and their complements, with trivial and impossible endpoint cases handled directly.

  For completeness, let~$M$ satisfy~$\Gamma$ and choose a labelled edge witnessing each nontrivial lower bound. When $E_i\cap B_j=\emptyset$, the degree of the witnessing edge lies outside~$B_j$, so either the successor value or the label cap belongs to~$B_j$. Make the choices accordingly. The resulting co-rule passes the label tests and all its conclusions are realized in~$M$.

  For soundness, take the designated witnesses~$x_i$ as in~\Cref{rem:designated-witnesses}. Choose each label~$m_i$ by the preceding test and choose the degree of the edge $(m_i,x_i)$ in $E_i\cap\bigcap_{j:\,E_i\cap B_j\ne\emptyset}B_j$. This intersection is nonempty because the intervals~$B_j$ are nested. Put all other edge degrees equal to~$0$. Each chosen edge satisfies its lower bound, and each upper bound is respected by the edge degree, the successor value required in~$\Delta_i$, or the label. If there are no nontrivial lower bounds, the co-rule has zero conclusions and the zero transition suffices. There are at most $2^{n^2}$ choices and at most~$n$ conclusions per co-rule.
\end{proof}

\begin{theorem}\label{metricpolynomialspace}
  Under the polynomial-space ball-intersection assumption, metric modal logic with either quantitative or crisp transitions is polynomial-space bounded over the corresponding finite-support functor.
\end{theorem}
\begin{proof}
  We can compute the conclusions of the co-rules from~\Cref{metricexpbranching} by enumerating the choices as in the quantitative case, using interval and ball-intersection tests to discard infeasible choices. The choices and counters have polynomial size, and the tests use polynomial space by assumption. Each endpoint in a conclusion is either $0$, $1$, or an endpoint of the premise.
\end{proof}

\begin{remark}
  In the crisp case, the finite-model reduction preserves the degree~$1$ of each selected edge, and the co-rule tests use only that degree on witnessing edges. Thus both arguments apply to the finite powerset subfunctor.
\end{remark}

\end{document}